\documentclass[a4paper,unpublished,onecolumn,11pt]{quantumarticle}
\pdfoutput=1

\usepackage{amsmath}
\usepackage{mathtools}
\usepackage{amsthm}
\usepackage{thmtools, thm-restate}
\usepackage{physics}    
\usepackage{amssymb}
\usepackage{amsfonts}
\usepackage{mathrsfs}   
\usepackage{nicefrac}   
\usepackage{stmaryrd}

\usepackage{xcolor}

\usepackage{xspace}

\usepackage{enumitem}

\usepackage{algorithm}
\usepackage{algorithmicx}
\usepackage[noend]{algpseudocode}

\usepackage{tikz}
\usetikzlibrary{decorations.markings}
\usepackage{tikzit}
\usepackage{float}

\tikzstyle{blank}=[fill=white, draw=white, shape=circle, inner sep=0, minimum size=.3cm]
\tikzstyle{black_circle}=[fill=black, draw=black, shape=circle]
\tikzstyle{white_circle}=[fill=white, draw=black, shape=circle]
\tikzstyle{blue_circle}=[fill=blue, draw=black, shape=circle]
\tikzstyle{black_square}=[fill=black, draw=black, shape=rectangle]
\tikzstyle{blue_square}=[fill=blue, draw=black, shape=rectangle]
\tikzstyle{grey square}=[fill=gray, draw=gray, shape=rectangle]
\tikzstyle{ctrl}=[fill=black, draw=black, shape=circle, inner sep=0, minimum size=1mm]
\tikzstyle{target inner}=[draw=black, shape=circle, cross out, rotate=45, inner sep=0, minimum size=1mm, tikzit fill={rgb,255: red,255; green,64; blue,47}]
\tikzstyle{target outer}=[draw=black, shape=circle, inner sep=0, minimum size=2mm, tikzit fill={rgb,255: red,255; green,51; blue,24}]
\tikzstyle{gate}=[fill=white, draw=black, shape=rectangle]
\tikzstyle{white_borderless}=[fill=white, draw=white, shape=circle, scale=0.8]

\tikzstyle{blue_dashed_arrow_right}=[->, draw=blue, dashed, thick]
\tikzstyle{black_arrow_left}=[<-]
\tikzstyle{blue_dashed_arrow_left}=[<-, draw=blue, dashed, thick]
\tikzstyle{red_dashed_arrow_left}=[draw={rgb,255: red,191; green,0; blue,64}, <-, dashed, thick]
\tikzstyle{red_dashed_arrow_right}=[draw={rgb,255: red,191; green,0; blue,64}, ->, dashed, thick]
\tikzstyle{new edge style 0}=[-, draw={rgb,255: red,144; green,144; blue,144}]
\tikzstyle{green_edge}=[->, draw={rgb,255: red,26; green,212; blue,79}, dashed, thick]
\tikzstyle{orange_edge}=[draw={rgb,255: red,255; green,176; blue,66}, ->, dashed, thick]
\tikzstyle{black_edge}=[->, thick]
\tikzstyle{new edge style 1}=[-, draw={rgb,255: red,150; green,150; blue,150}, dashed]
\tikzstyle{black_thick_l}=[thick, <-]
\tikzstyle{dotted blue}=[-, draw={rgb,255: red,93; green,195; blue,239}, dashed, thick]
\tikzstyle{new edge style 2}=[draw=red, ->]
\tikzstyle{new edge style 3}=[-, draw={rgb,255: red,230; green,54; blue,19}]
\tikzstyle{new edge style 4}=[->, dashed]
\tikzstyle{thick}=[-, thick]

\usepackage[
  style=alphabetic,
  sorting=nyt,
  backend=biber,
  maxnames=6,
  maxcitenames=2,
  url=false,
  doi=true,
  isbn=false,
  eprint=false
]{biblatex}

\usepackage{hyperref}
\hypersetup{
  colorlinks=true,
  linkcolor=cyan,
  citecolor=gray
}

\newcommand{\N}{\mathbb{N}}
\newcommand{\Z}{\mathbb{Z}}
\newcommand{\F}{\mathbb{F}}

\newcommand{\R}{\mathbb{R}}
\newcommand{\C}{\mathbb{C}}

\newcommand{\Fflow}{\(\mathbb{Z}_d\)-flow\xspace}
\newcommand{\Fflows}{\(\mathbb{Z}_d\)-flows\xspace}
\newcommand{\measp}[1]{\mathcal{M}(#1)}

\newcommand{\trans}[1]{{#1}^{\top}}
\newtheorem{theorem}{Theorem}
\newtheorem{definition}[theorem]{Definition}
\newtheorem{proposition}[theorem]{Proposition}
\newtheorem{lemma}[theorem]{Lemma}
\newtheorem{corollary}[theorem]{Corollary}

\title{Outcome determinism in measurement-based quantum computation with qudits}

\author{Robert I. Booth}
\affiliation{
  Sorbonne Universit\'e, CNRS, LIP6,
  4 place Jussieu, \mbox{F-75005} Paris, France
}
\affiliation{
  LORIA CNRS, Inria Mocqua, Universit\'e de Lorraine, \mbox{F-54000} Nancy,
  France
}

\author{Aleks Kissinger}
\affiliation{
  University of Oxford, United Kingdom
}

\author{Damian Markham}
\affiliation{
  Sorbonne Universit\'e, CNRS, LIP6,
  4 place Jussieu, \mbox{F-75005} Paris, France
}
\affiliation{
 JFLI, CNRS / National Institute of Informatics, University of Tokyo, Tokyo, Japan
}

\author{Cl\'ement Meignant}
\affiliation{
  Sorbonne Universit\'e, CNRS, LIP6,
  4 place Jussieu, \mbox{F-75005} Paris, France
}

\author{Simon Perdrix}
\affiliation{
  LORIA CNRS, Inria Mocqua, Universit\'e de Lorraine, \mbox{F-54000} Nancy,
  France
}

\begin{document}

\maketitle

\begin{abstract}
  In measurement-based quantum computing (MBQC), computation is carried out by a
  sequence of measurements and corrections on an entangled state.
  Flow, and related concepts, are powerful techniques for characterising the
  dependence of the corrections on previous measurement outcomes.
  We introduce flow-based methods for MBQC with qudit graph states, which we
  call \Fflow, when the local dimension is an odd prime.
  Our main results are a proof that \Fflow is a necessary and sufficient condition
  for a strong form of outcome determinism.
  Along the way, we find a suitable generalisation of the concept of
  measurement planes to this setting and characterise the allowed measurements
  in a qudit MBQC.
  We also provide a polynomial-time algorithm for finding an optimal \Fflow
  whenever one exists.
\end{abstract}

In measurement based quantum computation (MBQC), one starts with an entangled 
resource state (usually graph states \cite{briegel_persistent_2001}), and
computation is carried out by sequential measurements where at each stage the
measurement choice depends on previous results \cite{raussendorf_one-way_2001,
  raussendorf_computational_2002, danos_determinism_2006,
  danos_measurement_2007}. This adaptivity is necessary to combat the randomness
induced by measurements, and plays an important role both foundationally
\cite{da_silva_closed_2011, raussendorf_measurement-based_2011} and in terms of
trade-offs in optimizing computations \cite{broadbent_parallelizing_2009,
  miyazaki_analysis_2015}.

Causal flow \cite{danos_determinism_2006}, and its generalisation gflow
\cite{browne_generalized_2007}, are graph-theoretical tools for characterising
and analysing adaptivity for graph states in MBQC.
They have proven a powerful tool including optimizing adaptive measurement 
patterns \cite{eslamy_geometry-based_2018}, translating between MBQC and the
circuit picture \cite{miyazaki_analysis_2015, backens_there_2021} , with
application for parallelising quantum circuits \cite{broadbent_parallelizing_2009},
to construct schemes for the verification of blind quantum computation
\cite{fitzsimons_unconditionally_2017, mantri_flow_2017}, to extract bounds on
the classical simulatability of MBQC \cite{markham_entanglement_2014}, to prove
depth complexity separations between the circuit and measurement-based models of
computation \cite{broadbent_parallelizing_2009, miyazaki_analysis_2015}, to
study trade-offs in adiabatic quantum computation \cite{antonio_adiabatic_2014}
and recently for applying ZX-calculus techniques
\cite{duncan_graph-theoretic_2020, de_beaudrap_pauli_2020}, including to circuit
compilation \cite{duncan_graph-theoretic_2020}.

In recent years, we have seen increased interest in computing, and quantum
information in general, over higher dimensional qudit systems, as opposed to
qubits \cite{wang_qudits_2020}. The added flexibility of increased dimension
allows, for example, for shorter circuits in computation
\cite{kiktenko_scalable_2020}, asymptotic improvement in circuit depths
\cite{gokhale_asymptotic_2019}, optimal error correcting codes
\cite{cleve_how_1999} and noise tolerance in quantum key distribution
\cite{bechmann-pasquinucci_quantum_2000, cerf_security_2002}. Furthermore many
physical systems exist which naturally encode qudits \cite{blatt_entangled_2008,
  erhard_twisted_2018, gao_arbitrary_2019}. This has motivated the translation
of MBQC into qudits \cite{zhou_quantum_2003}, which naturally leads to the
question: can the flow techniques above be extended to the qudit setting?

In this work we generalise gflow to the qudit setting, when the local dimension 
is an odd prime. In section \ref{sec:preliminaries} we review quantum
computation with qudits and introduce our computational model, a qudit version
of the measurement calculus \cite{danos_measurement_2007}. This requires careful
consideration of the measurements which are allowed as part of the MBQC. We also
discuss the various determinism conditions present in the literature and in
particular define robust determinism. In section \ref{sec:flow} we introduce
\Fflow, and show that it is sufficient to obtain a robustly deterministic MBQC.
We also prove a converse: any robustly deterministic MBQC is shown to have a
\Fflow. Finally, in section \ref{sec:algorithm} we present a polynomial-time
algorithm for determining if a given graph has a \Fflow, and further prove that
it always produces \Fflows of minimal depth (if it succeeds). These results are
the first step in a characterisation of outcome determinism in MBQC for a large
class of finite-dimensional quantum systems.

\section{Preliminaries}
\label{sec:preliminaries}
Throughout this paper, \(d\) denotes an arbitrary prime, and \(\Z_d = \Z/d\Z\)
the ring of integers with arithmetic modulo \(d\).
We also put \(\omega \coloneqq e^{i\frac{2\pi}{d}}\), and let \(\Z_d^*\) be the
group of units of \(\Z_d\).
Since \(d\) is prime, \(\Z_d\) is a field and \(\Z_d^* = \Z_d \setminus \{0\}\)
as a set.

\subsection{Computational model}

The Hilbert space of a qudit
\cite{gottesman_fault-tolerant_1999,wang_qudits_2020} is \(\mathcal{H} =
\operatorname{span}\{\ket{m} \mid m \in \Z_d\} \cong \C^d\), and we write
\(U(\mathcal{H})\) the group of unitary operators acting on \(\mathcal{H}\). We
have the following standard operators on \(\mathcal{H}\), also known as the
clock and shift operators:
\begin{equation}
  Z \ket{m} \coloneqq \omega^{m} \ket{m} \qand
  X \ket{m} \coloneqq \ket{m+1}
  \qq{for any} m \in \Z_d.
  \label{eq:pauli} 
\end{equation}
In particular, note that \(ZX = \omega XZ\).
We call any operator of the form \(\omega^k X^a Z^b\) for \(k,a,b \in \Z_d\) a
\emph{Pauli operator}, although we will often drop the phase \(\omega^k\) as it
is of little importance in most cases. We say a Pauli operator is \emph{trivial}
if it is proportional to the identity.  The Paulis are further related by the
Hadamard gate:
\begin{equation} 
  H\ket{m} = \frac{1}{\sqrt{d}} \sum_{n \in \Z_d} \omega^{mn} \ket{n}
  \qq{s.t.} HXH^\dagger = Z \qand HZH^\dagger = X^{-1}.
  \label{eq:hadamard}
\end{equation}
Equations \eqref{eq:pauli} and \eqref{eq:hadamard} imply that both \(X\) and
\(Z\), and in fact every Pauli has spectrum \(\{\omega^k \mid k \in \Z_d\}\).

We also use the controlled-Z gate, which acts on \(\mathcal{H} \otimes
\mathcal{H}\),
\begin{equation}
  E \ket{m} \ket{n} \coloneqq \omega^{mn} \ket{m}\ket{n}.
\end{equation}

It is important to emphasise a key difference between the qudit and the qubit
case: when \(d \neq 2\), none of these operators are self-inverse.
In fact, if \(Q\) is a Pauli and \(I\) the identity operator on \(\mathcal{H}\),
we have:
\begin{equation}
  Q^d = I, \quad E^d = I \otimes I \qand H^4 = I.
\end{equation}
As a result, they are not self-adjoint either, something which needs to be taken
into account when describing measurements.

\subsubsection{Measurement spaces}
\label{ssec:measurement_spaces}

For qubit MBQC, it is well established that by using a Pauli \(X\), \(Y\) or
\(Z\) as an acausal correction operator, it is possible to perform MBQC on graph
states where the measurements are taken from the plane on the Bloch sphere
orthogonal to the correction \cite{browne_generalized_2007}. Since there are
three Pauli operators for qubits, this yields three allowable measurement planes
for MBQC.

This interpretation is not as clear in the qudit case, partly because Pauli
operators are self-adjoint only in the case \(d=2\), but mostly because the
geometry of the Bloch ``space'' is not as intuitive in the general case.
A qudit-measurement will be described by a unitary matrix $M$: given its
spectral decomposition $M=\sum_i \lambda_i P_i$, an $M$-measurement is the
projective measurement $\{P_i\}_{i \in \Z_d}$. In the context of MBQC, we would
like to have a distinguished measurement outcome, the one that does not need
corrections, so we assume that all measurements have a fixpoint.

\begin{definition}
  $M$ is a \emph{fixpoint unitary} if $M^\dagger M = M M^\dagger = I$ and
  $\exists \ket \phi \neq 0$ s.t. $M\ket \phi =\ket \phi$. Given $(a,b) \in
  \Z_d^2 \setminus \{(0,0)\}$, the \emph{measurement space} $\mathcal{M}(a,b)$
  is defined as
   \(\mathcal{M}(a,b):=\{\text{fixpoint unitaries }M\text{~s.t.~}X^aZ^bM=\omega
    MX^aZ^b\}\). 
\end{definition}

It should be pointed out that the commutation relation used to define the
measurement space \(\mathcal{M}(a,b)\) is somewhat arbitrary.
We could have chosen instead to use the relation
\begin{equation}
  X^a Z^b M = \omega^p M X^a Z^b \qq{for some} p \in \Z_d^*.
  \label{eq:measurement_commutation}
\end{equation}
However, nothing is lost by considering only \(p=1\), since if \(M\) verifies
equation \eqref{eq:measurement_commutation}, then
\begin{equation}
  X^{p^{-1}a} Z^{p^{-1}b} M = \omega^{p^{-1}p} M X^{p^{-1}a} Z^{p^{-1}b}
  = \omega M X^{p^{-1}a} Z^{p^{-1}b},
\end{equation}
(where this calculation is formally carried out using \(p^{-1} =
\frac{d+1}{p}\)) which implies that \(M \in \mathcal{M}(p^{-1}a,p^{-1}b)\).

In fact, this construction is very analogous to one used in qudit quantum error
correction where the \(M\) are called detectable errors
\cite{gottesman_fault-tolerant_1999}.
The main point of this definition is that the Pauli \(X^a Z^b\) can be used to
translate the eigenvectors of any measurement in the corresponding measurement
space:
\begin{proposition}
   If \(M \in \mathcal{M}(a,b)\) for some non-trivial Pauli \(Q=X^a Z^b\), then
  the spectrum of \(M\) is \(\{\omega^m \mid m \in \Z_d\}\), each eigenvalue
  has multiplicity \(1\), and \(M\) is special unitary.
  Denoting \(\ket{0:M}\) the fixpoint of \(M\), then \(\ket{m:M} = Q^{-m}
  \ket{0:M}\) is an eigenvector of \(M\) associated with eigenvalue
  \(\omega^m\).
  \label{prop:measurement_PVM}
\end{proposition}
\begin{proof}
  By assumption, if \(M \in \mathcal{M}(Q)\) then \(M\) has a fixpoint \(M
  \ket{0:M} = \ket{0:M}\).
  Then, it follows from the commutation relation that
  \begin{equation}
    MQ\ket{0:M} = \omega^{-1} QM \ket{0:M} = \omega^{-1} Q \ket{0:M},
  \end{equation}
  so \(Q\ket{0:M}\) is an
  eigenvector of \(M\) associated with eigenvalue \(\omega^{-1}\).
  Repeating this procedure, we find that \(\ket{k:M} = Q^{-k}\ket{0:M}\)
  is an eigenvector of \(M\) associated with eigenvalue \(\omega^k\), and a
  counting argument shows that each of these eigenvalues must have multiplicity
  \(1\).
  Now, we have that \(\det(M) = \prod_{k \in \Z_d} \omega^k = 1\).
\end{proof}

This means that the Pauli \(Q\) can be used as correction for any measurement in
the corresponding measurement space, as is described in section \ref{sec:flow}.
As in the qubit case, pairs of measurements within the same measurement space
\(\mathcal{M}(a,b)\) are still related to each other by rotations around the
``correction'' axis \(X^a Z^b\):
\begin{proposition}
  Let \(Q=X^a Z^b\) be a non-trivial Pauli operator and \(N \in
  \mathcal{M}(a,b)\).
  Then \(M \in \mathcal{M}(a,b)\) if and only if there is a special unitary \(U
  \in SU(d)\) such that \(M = UNU^\dagger\) and \([U,Q] = 0\).
  \label{prop:measurement_space_unitary}
\end{proposition}
\begin{proof}
  \noindent\((\implies)\) If \(M \in \mathcal{M}(Q)\) then \(\mathrm{sp}(M) =
  \mathrm{sp}(N) = \{\omega^k \mid k \in \mathbb{Z}_d\}\) and each eigenvalue
  has multiplicity one. It follows that \(M\) and \(N\) are similar so that
  there is a unitary \(U\) such that \(M = UNU^\dagger\).

  Furthermore, by proposition \ref{prop:measurement_PVM}, the eigenvector
  \(\ket{k:M}\) of \(M\) can be obtained as \(Q^{-k}\ket{0:M}\), from which it
  also follows that \(\ket{k+1:M} = Q\ket{k:M}\).
  But, we also have \(MU\ket{k+1:N} = UNU^\dagger U \ket{k+1:N} = \omega^{k+1}
  U\ket{k+1:N}\) from which it follows that
  \begin{equation}
    Q U \ket{k:N} = Q \ket{k:M} = \ket{k+1:M} = U \ket{k+1:N} = U Q \ket{k:N}.
  \end{equation}
  This is true for any \(k\in\mathbb{Z}_d\), and since \(N\) is unitary its
  eigenvectors form a basis for \(\mathcal{H}\). We deduce that \(QU = UQ\).

  Finally, it is clear we can choose \(U\) to be special unitary, since for any
  unit norm \(\lambda \in \mathbb{C}\), \((\lambda U)N(\lambda U)^\dagger =
  \abs{\lambda}^2 UNU^\dagger = UNU^\dagger\).

  \noindent\((\impliedby)\) Let \(M = UNU^\dagger\) such that \([U,Q] = 0\),
  then we have
  \begin{equation}
    MQ = UNU^\dagger Q = UNQU^\dagger = \omega UQNU^\dagger
    = \omega QUNU^\dagger = \omega QM.
  \end{equation}
  Furthermore, \(M\) and \(N\) have the same spectrum, and in particular \(M\)
  has a fixpoint since \(N\) does. Then, \(M \in \mathcal{M}(Q)\).
\end{proof}

In turn, this allows us to recover a parametrisation of measurement spaces much
closer to the qubit case, where a measurement is given by angles relative to a
reference Pauli axis of the Bloch sphere.
\begin{corollary}[Measurement angles]
  For any non-zero \((a,b) \in \Z_d^2\), a measurement \(M \in \mathcal{M}(a,
  b)\) is characterised by $d-1$ angles $\vec \theta = (\theta_1, \ldots,
  \theta_{d-1})\in [0,2\pi)^{d-1}$, up to a choice of reference axis \(P \in
  \mathcal{M}(a, b)\).
  \label{cor:measurement_angles}
\end{corollary}
\begin{proof}
  Fix some \(P \in \mathcal{M}(a,b)\), then by the proposition, every \(M \in
  Q^\dagger\) is such that \(M = UPU^\dagger\), and in particular \([U,Q] = 0\).
  This implies that in the eigenbasis of \(Q\), \(U\) takes the form of a
  diagonal matrix \(\mathrm{diag}(e^{i \theta_k} \mid k \in \mathbb{Z}_d)\) with
  \(\theta_k \in [0,2\pi)\).
  Since \(\det(U) = 1\), we have that \(\sum_{k=0}^{d-1} \theta_k = 0\) and one
  of these phases is redundant.
  Then, \(U\) and by extension, \(M\), is uniquely determined by the \(d-1\)
  phases \(\{\theta_k\}_{k=1}^{d-1}\) (and the arbitrary choice of \(P\)).
\end{proof}

As is the case for qubits, the choice of reference axes (one per measurement
space) is entirely arbitrary, so we assume for the rest of the article that
some fixed choice has been made for each measurement space.

\subsubsection{Measurement patterns}

Given that we are interested in procedures with an emphasis on measurements and
corrections conditioned on the outcomes of measurements, the quantum circuit
description of computations is not very practical for our needs.
Instead, we describe an MBQC by a sequence of commands, called a measurement
pattern. The description of measurements hinges on the characterisation of
measurement spaces in corollary \ref{cor:measurement_angles}. We suppose some
arbitrary choice of reference is made for each measurement space, then:
\begin{definition}[\cite{danos_measurement_2007}]
  A \emph{measurement pattern} on a register \(V\) of qudits consists in a
  finite sequence of \(V\)-indexed commands chosen from:
  \begin{itemize}
  \item $N_u$ : initialisation of a qudit $u$ in the state $\ket{0:X} =
    H\ket{0}$;
  \item $E_{u,v}^\lambda$ : application of $E^\lambda$ on qudits $u$ and
    $v$ for some \(\lambda \in \Z_d\), with \(u \neq v\);
  \item $M_u^{a,b}(\vec{\theta})$ : measurement of qudit $u$ in the
    measurement space \(\mathcal{M}(a,b)\) with angles \(\vec{\theta}\);
  \item $X_u^{m_v}$ and $Z_u^{m_v}$ : Pauli corrections depending on the outcome
    \(m_v\) of the measurement of vertex \(v\).
  \end{itemize}

  A measurement pattern is \emph{runnable} if no commands act on non-inputs
  before they are initialised (except initialisations) or after they are
  measured, and no commands depend on the outcome of a measurement before it is
  made.
\end{definition}

\subsection{A graph-theoretical representation}

Following \textcite{klasing_determinism_2017}, we show that measurement patterns
can be equivalently represented as labelled graphs.
Then, following \cite{zhou_quantum_2003}, these measurement patterns are
universal for all qudit quantum circuits.

\paragraph{Notation.}
A $\Z_d$-graph $G$ is a loop-free undirected $\Z_d$-edge-weighted
graph on a set $V$ of vertices.
We will identify the graph $G$ with its symmetric adjacency matrix $G \in
\Z_d^{V\times V}$ (for some arbitrary ordering of the rows and columns).
If \(A,B \subseteq V\), we will also denote \(G[A,B]\) the submatrix of \(G\)
obtained by keeping only the rows corresponding to elements of \(A\) and the
columns corresponding to elements of \(B\).
If \(A \subset V\), then we denote \(1_A \in \Z_d^V\) the column vector whose
\(u\)-th element is \(1\) if \(u \in A\), \(0\) otherwise.
Similarly we consider $\Z_d$-multisets of vertices where each vertex occurs with
a multiplicity in $\Z_d$ and we will identify the $\Z_d$-multiset with column
vectors in $\Z_d^{V}$.
The size of a multiset is defined by $|A| = \sum_{u\in V} A(u) \in Z_d$. $\trans
x$ is the transpose of $x$.
Given a Pauli operator $P$ and a multiset $A$, let $P_A := \bigotimes_{u\in
  V}P_u^{A(u)}$.
\vspace{1mm}

The commands of a measurement pattern verify the following identities, for every
\(u,v \in V\) such that \(u \neq v\):
\begin{align}
  X_u Z_v &= Z_v X_u, &
  X_u Z_u &\simeq Z_u X_u, \label{eq:commutation_rels_1} \\
  X_u M_v &= M_v X_u, &
  Z_u M_v &= M_v Z_u, \\
  E_{u,v} X_u &= X_u Z_v E_{u,v}, &
  E_{u,v} Z_u &= Z_u E_{u,v}. \label{eq:commutation_rels_3}
\end{align}
where we use the notation \(A \simeq B\) to mean that there is a phase
\(e^{i\alpha}\) such that \(A = e^{i\alpha} B\).
It was shown by \textcite{danos_measurement_2007} that any runnable measurement
pattern can be rewritten using these commutation relations to the standard
form\footnote{They worked in the qubit setting but their proof is purely
  symbolic. Rewriting \(U^\lambda = \prod_{k=0}^\lambda U\) where \(U\)
  is any unitary from equations
  \eqref{eq:commutation_rels_1}-\eqref{eq:commutation_rels_3}, and applying
  their standardisation procedure results in a pattern of the form
  \eqref{eq:standard_form}.}:
\begin{equation}
  \mathfrak{P} \simeq \left( \prod_{v \in O^\mathsf{c}}^\prec X_{\vb{x}(v)}^{m_v} Z_{\vb{z}(v)}^{m_v} M_v^{a_v,b_v}(\vec{\theta}_v) \right)
  \left( \prod\limits_{(u,v) \in G} E_{u,v}^{G_{uv}} \right)
  \left( \prod_{v \in I^\mathsf{c}} N_v \right),
  \label{eq:standard_form}
\end{equation}
where \(\vb{x},\vb{z}\) are functions \(O^\mathsf{c} \to \Z_d^V\),
\(m_v\) is the outcome of the measurement \(M_v\), \(G\) is the adjacency matrix
of a \(\Z_d\)-graph, and \((u,v) \in G\) identifies an edge in the graph \(G\).
The functions \(\vb{x},\vb{z}\) implicitly describe a measurement
order: the transitive closure of the relation \(\{(u,v) \mid \vb{x}(v)_u \neq 0
\text{ or } \vb{z}(v)_u \neq 0\}\) gives a strict partial order \(\prec\) on
\(O^\mathsf{c}\). The measurement order must agree with \(\prec\) if the pattern
is runnable.

This motivates the following definition \cite{danos_determinism_2006,
  browne_generalized_2007, backens_there_2021}:
\begin{definition}
  An \emph{open \(\Z_d\)-graph} is a triple \((G,I,O)\) where $G$ is a
  $\Z_d$-graph over $V$, and $I,O\subseteq V$ are distinguished sets of vertices
  which identify inputs and outputs in an MBQC.

  A \emph{labelled open \(\Z_d\)-graph} is a tuple \((G,I,O,\lambda)\) where
  \((G,I,O)\) is an open \(\Z_d\)-graph and $\lambda : O^c \to \Z_d^2 \setminus
  \{(0,0)\}$ assigns a measurement space to each measured vertex.
\end{definition}

Given the form \eqref{eq:standard_form}, it is clear that we can describe a
runnable measurement pattern by a tuple
\((G,I,O,\lambda,\vb{x},\vb{z},\vb{M})\), where \((G,I,O,\lambda)\) is a
labelled open \(\Z_d\)-graph, \(\vb{x},\vb{z}\) are functions \(O^\mathsf{c} \to
\Z_d^V\), and \(\vb{M}\) is a function \(O^\mathsf{c} \to U(\mathcal{H})\) such
that \(\vb{M}(u) \in \mathcal{M}(\lambda(u))\) for all \(u \in O^\mathsf{c}\).
\(\vb{M}\) gives the measurement to be made at each non-output vertex, and
\(\vb{x},\vb{z}\) describe corresponding outcome-dependant corrections. The
labelling \(\lambda\) is technically required since the syntax of measurements
in equation \eqref{eq:standard_form} depend on the labelling, but as we shall
see in the next section, once the choice of \(\vb{M}\) is made, \(\lambda\) has
no effect on the actual computation carried out by the MBQC (the semantics of
the measurement pattern).

\subsection{Determinism}

An MBQC \((G,I,O,\lambda,\vb{x,z,M})\) describes an inherently probabilistic
computation with \(d \times \abs{O^\mathsf{c}}\) possible branches (one for each
set of measurement outcomes). Given an input state \(\ket{\phi} \in
\mathcal{H}^{\otimes I}\) and set of outcomes \(\vec{m} \in
\Z_d^{O^\mathsf{c}}\), the corresponding branch is given by:
\begin{equation}
  A_{\vec{m}}(\ket{\phi}) \coloneqq
  \left( \prod_{v \in O^\mathsf{c}}^\prec X_{\vb{x}(v)}^{m_v} Z_{\vb{z}(v)}^{m_v} \bra{m_v:\vb{M}(v)}_v \right)
  \left( \prod\limits_{(u,v) \in G} E_{u,v}^{G_{uv}} \right)
  \left(\ket{\phi} \bigotimes_{u\in I^\mathsf{c}} \ket{0:X} \right)
  \label{eq:branch_map}
\end{equation}
The branch maps give a Kraus decomposition for the CPTP map
\(\mathcal{H}^{\otimes I} \to \mathcal{H}^{\otimes O}\) implemented by the MBQC:
\begin{equation}
  \rho \longmapsto \sum_{\vec{m} \in \Z_d^{O^\mathsf{c}}} A_{\vec{m}} \rho A_{\vec{m}}^\dagger.
\end{equation}

A measurement pattern is said to be \emph{deterministic} if the output does not
depend on the outcomes of the measurements. This is equivalent to saying that
all branches \eqref{eq:branch_map} are proportional, in which case the pattern
is described by the single Kraus operator \(K_{\vec{0}}\), corresponding to
obtaining outcome \(0\) for all measurements. This is by construction a
correction-less branch since we have then obtained the ``preferred'' outcome of
each measurement.  However, a problem comes up if \(K_{\vec{0}} = 0\), in which
case two deterministic MBQCs can have the same open graph but implement
different maps. See \cite{danos_determinism_2006, browne_generalized_2007,
  klasing_determinism_2017} for examples.

To exclude these pathological cases, a stronger determinism condition was
introduced by \textcite{danos_determinism_2006}: a measurement pattern is
\emph{strongly  deterministic} if all branch maps are equal up to a global
phase. In particular, strongly deterministic measurement patterns implement
isometries. 

Now, the original purpose of flow was to obtain sufficient and necessary
conditions for deciding when such an MBQC is deterministic. However, a
characterisation of strong determinism is still an open question, even in
the case of qubits. Instead, we restrict our attention to a yet stronger form of
determinism, which is both more tractable and arguably more practical
\cite{klasing_determinism_2017}:
\begin{definition}[Robust determinism]
   \((G,I,O,\lambda,\vb{x},\vb{z})\) is \emph{robustly deterministic} if
   for any \mbox{\(\prec\)-lowerset}\footnote{If \(\prec\) is a partial order on
   \(V\), then a \(\prec\)-lowerset is a subset \(S \subseteq V\) such that if
   \(u \prec v\) for some \(v \in S\), then \(u \in S\).} \(S \subseteq
   O^\mathsf{c}\) and any \(\vb{M} : S \to U(\mathcal{H})\) such that
   \(\vb{M}(u) \in \mathcal{M}(\lambda(u))\), the MBQC \((G,I,O \cup
   S^\mathsf{c},\lambda|_S,\vb{x}|_S,\vb{z}|_S,\vb{M})\) is strongly
   deterministic.
\end{definition}

Robust determinism is equivalent to the uniformly and stepwise strong
determinism of \textcite{browne_generalized_2007} in the qubit case.

\subsection{Graph states}

For an open graph \((G,I,O)\) and an arbitrary input state \(\ket{\phi}
\in \mathcal{H}^{\otimes I}\), we write
\begin{equation}
  \ket{G(\phi)} = \left(\prod\limits_{\substack{u,v\in V \\ u<v}} E_{u,v}^{G_{u,v}}
  \right)\left(\ket{\phi} \bigotimes_{u\in I^\mathsf{c}} \ket{0:X}_u\right),
\end{equation}
which we call an open graph state. Open graph states are resource states for the
MBQCs which we describe in this paper. In the case \(I=\varnothing\), we recover
the well-known qudit graph states \cite{zhou_quantum_2003, marin_access_2013}.
The stabilisers of an open graph state are given by:
\begin{proposition}[Open graph stabilisers]
  Let \((G,I,O)\) be an open graph, and \(Q\) a product of Paulis.
  Then, \(Q \ket{G(\phi)} = \ket{G(\phi)}\) for all \(\ket{\phi} \in
  \mathcal{H}^{\otimes I}\) if and only if there is a multiset \(A \in \Z_d^V\)
  such that \(A_v = 0\) for all \(v \in I\) and \(Q = \omega^{\frac{\trans
      AGA}{2}} X_A Z_{GA}\).
  \label{prop:open_graph_stabilisers}
\end{proposition}
\begin{proof}
  The stabilisers of \(\ket{G(\phi)}\) are simply the stabilisers of \(\ket{\phi}
  \bigotimes_{u \in I^\mathsf{c}} \ket{0:X}\) conjugated by \(E_G\).
  It is clear that the stabiliser group of \(\ket{0:X}\) is generated by
  \(X^m\) for all \(m \in \Z_d\).
  Since no Pauli stabilises every \(\ket{\psi} \in \mathcal{H}\), it follows
  that the stabiliser group of \(\ket{\phi} \bigotimes_{u \in I^\mathsf{c}}
  \ket{0:X}\) is of the form \(X_A\) for some \(A \in \Z_d^{V}\) such that \(A_v
  = 0\) if \(v \in I\).
  Now, we have
  \begin{align}
    E_G X_k^{A_k} E_G^\dagger
    = \prod\limits_{\substack{u,v\in V \\ u<v}} E_{u,v}^{G_{u,v}}  X_k^{A_k}
    \prod\limits_{\substack{u,v\in V \\ u<v}} E_{u,v}^{G_{u,v}} 
    = X_k^{A_k} \prod_{v \in V} Z_v^{G_{vk}A_k} = X_{A_k} Z_{GA_k}
  \end{align}
  so that
  \begin{align}
    E_G X_A E_G^\dagger
    &= E_G \prod_{k \in V} X_k^{A_k} E_G^\dagger
    = X_k^{A_k} \prod_{v \in V} Z_v^{G_{vk}A_k}\\
    &= \prod_{k \in V} X_{A_k} Z_{GA_k}
    = \omega^{\sum_{k \in V}A_k^\mathsf{T}G\sum_{k \in V}A_k} X_{\sum_{k \in V} A_k} Z_{\sum_{k \in V}GA_k}\\
    &= \omega^{A^\mathsf{T}GA} X_A Z_{GA}
  \end{align}
  as claimed.
\end{proof}

\section{\Fflow and outcome determinism}
\label{sec:flow}
This leads us to the statement of our novel flow condition.
It is a strict generalisation of gflow to qudits, which must take into
account the additional freedom in open graphs described above.
\begin{definition}
  $(G,I,O,\lambda)$ has an \emph{\Fflow} \((C,\Lambda)\) if \(C\) is a matrix in
  \(\Z_d^{V\times V}\) and \(\Lambda\) a totally ordered partition of \(V\)
  such that
  \begin{enumerate}
  \item $\forall u \in O^c$, $\lambda(u)=(C_{uu},(GC)_{uu})$;
  \item \(C[I,V] = 0\) and \(C[V,O] = 0\);
  \item for any \(M,N \in \Lambda\),
    \begin{itemize}
    \item \(C[M,M]\) and \((GC)[M,M]\) are diagonal;
    \item whenever \(M < N\), \(C[M,N] = (GC)[M,N] = 0\).
    \end{itemize}
  \end{enumerate}
  We call \(\Lambda\) a \emph{layer decomposition} of \((G,I,O,\lambda)\) for
  \(C\) and the elements of \(\Lambda\) are \emph{layers}.
\end{definition}

If \((G,I,O,\lambda)\) is a labelled open graph with \Fflow \((C,\Lambda)\),
then we obtain a runnable MBQC \((G,I,O,\lambda,\vb{x}_C,\vb{z}_C)\) by imposing
\begin{equation}
  \vb{x}_C(v) \coloneqq (C_{\bullet v} - \lambda(v)_1 1_{\{v\}})
  \qand \vb{z}_C(v) \coloneqq ((GC)_{\bullet v} - \lambda(v)_2 1_{\{v\}}),
  \label{eq:corrections_flow}
\end{equation}
where \(M_{\bullet v}\) is the \(v\)-th column of a matrix \(M\).

The layer decomposition \(\Lambda\) describes a (partial) measurement
order for the non-output qudits: the qudits can be measured in any
totalisation of the order induced on \(O^\mathsf{c}\) by the order of
\(\Lambda\), and qudits within the same layer can be measured simultaneously.
This order is a (not necessarily strict) extension of the order \(\prec\)
induced by \(\vb{x}_C,\vb{z}_C\) as described in the previous section.

The \(u\)-th columns (minus the \(u\)-th element) of \(C\) and \(GC\) describe
where to apply \(X\) and \(Z\) corrections for the measurement of vertex \(u \in
V\), respectively.
The elements in the \(u\)-th columns above \(u\) then correspond to qudits that
have already been measured, and must be zero for there to be no unwanted
back-action.
The elements in the \(u\)-th columns that corresponds to vertices in the same
layer as \(u\) must also be \(0\), since those vertices can be measured before
\(u\).
These considerations impose condition (iii).

Condition (ii) follows from the fact that the outputs are not measured and thus
no correction is needed.
Furthermore, we cannot apply \(X\) corrections at an input vertex, since the
measurement pattern we introduce relies on the fact that \(X_u N_u = N_u\).
Finally, the \(u\)-th element of the \(u\)-th column describes what correction
is applied at vertex \(u\), \(X^C_{uu} Z^{(GC)_{uu}}\), when we follow this
procedure.
This correction must match the measurement space assigned to \(u\) so that we
can use the back-action to perform the correction, which implies condition (i).

Then this MBQC is deterministic and implements an isometry:
\begin{theorem}
  Suppose the graph state \((G,I,O,\lambda)\) has \Fflow \((C,\Lambda)\), then
  the MBQC \((G,I,O,\lambda,\vb{x}_C,\vb{z}_C)\) is runnable and robustly
  deterministic. Furthermore, for a given choice of measurements \(\vb{M}\), it
  realises the isometry
  \begin{equation}
    \begin{aligned}
      \mathcal{H}^{\otimes I} &\longrightarrow \mathcal{H}^{\otimes O} \\
      \ket{\phi} &\longmapsto \bigotimes_{u \in O^\mathsf{c}} \bra{0 :
        \vb{M}(u)} E_G \left(\ket{\phi} \bigotimes_{u\in I^\mathsf{c}} \ket{0:X}
      \right)
    \end{aligned}.
  \end{equation}
  \label{thm:fflow_determinism}
\end{theorem}
\begin{proof}
  Assume $(G,I,O,\lambda)$ has a \Fflow $(C,\Lambda)$. We perform the
  measurements in the order given by any totalisation of the order induced by
  \(\Lambda\) on \(V\). We measure qudit $u$ with a
  $M$-measurement, and we obtain a classical outcome $s_u\in \Z_d$. Let
  $Q_u:= X_u^{C_{uu}}Z_u^{(GC)_{uu}}$ then by lemma
  \ref{prop:measurement_space_unitary}, the action of any measurement in
  \(\measp{\lambda(u)}\) correspond to the application on qudit $u$ of the
  projector $\bra {m:M} = \bra{0:M}Q_u^{s_u}$. Thus a correction must consist in
  simulating the application of $Q^{-s_u}$ on $u$. The definition of \Fflow
  implies that \(C\) and \(GC\) must be lower triangular, so that $X_{(C\{u\})\setminus
    \{u\}}Z_{(GC\{u\})\setminus \{u\}}$ acts only on unmeasured qudits,
  where $A\setminus \{u\}$ removes all the occurrences of $u$ in $A$:
  \begin{equation}
    A\setminus \{u\} = v \mapsto
    \begin{cases}
      0 &\qif u=v\\
      A(v)&\text{otherwise}
    \end{cases}.
  \end{equation}
  Then we have that:
  \begin{align}
   X^{s_u}_{(C\{u\})\setminus
    \{u\}}Z^{s_u}_{(GC\{u\})\setminus \{u\}} \ket G
    &= X^{s_u}_{(C\{u\})\setminus \{u\}}Z^{s_u}_{(GC\{u\})\setminus \{u\}}Q_u^{s_u}Q_u^{-s_u} \ket G \\
    &= X^{s_u}_{C\{u\}}Z^{s_u}_{GC\{u\}}Q_u^{-s_u} \ket G = Q_u^{-s_u} \ket G.
  \end{align}
  As a consequence, the correction $X^{s_u}_{(C\{u\})\setminus
    \{u\}}Z^{s_u}_{(GC\{u\})\setminus \{u\}}$ is runnable and makes the
  computation uniformly deterministic.

  Since all the branch maps are equal, the computation is strongly
  deterministic, and since we have considered only a single measurement and the
  associated corrections, it is stepwise deterministic.

  In \cite{danos_determinism_2006} it was shown that if a measurement pattern is
  strongly deterministic then it implements an isometry.
  Since we correct each measurement to the outcome \(m = 0\), it is clear that
  the final isometry is given by 
  \(\mathcal{H}^{\otimes I} \to \mathcal{H}^{\otimes O} : \prod_{u \in
    O^\mathsf{c}} \bra{0 : M_u} E_G N_{I^\mathsf{c}}\) as claimed.
\end{proof}

\begin{figure}
  \centering
  \begin{tikzpicture}
	\begin{pgfonlayer}{nodelayer}
		\node [style={black_square}] (0) at (-8.25, 1.25) {};
		\node [style={black_circle}] (1) at (-8.25, -0.75) {};
		\node [style={white_circle}] (2) at (-6.25, 1.25) {};
		\node [style={white_circle}] (3) at (-6.25, -0.75) {};
		\node [style={white_circle}] (4) at (-6.25, 1.25) {};
		\node [style=none] (5) at (-7.25, 1.75) {a};
		\node [style={white_borderless}] (6) at (-7.25, 0.25) {};
		\node [style=none] (7) at (-7.25, -1.25) {c};
		\node [style=none] (8) at (-7.25, 0.25) {b};
		\node [style=none] (9) at (6, 0) {\(\underbrace{  \begin{pmatrix} 0 & 0 & a & 0 \\ 0 & 0 & b & c \\ a & b & 0 & 0 \\ 0 & c & 0 & 0 \end{pmatrix}}_{G}   \underbrace{\begin{pmatrix} 0 & 0 & 0 & 0 \\ 0 & 1 & 0 & 0 \\ a^{-1} & 0 & 0 & 0 \\ 0 & 0 & 0 & 0 \end{pmatrix} }_{C} =  \begin{pmatrix} 1 & 0 & 0 & 0 \\ ba^{-1} & 0 & 0 & 0 \\ 0 & b & 0 & 0 \\ 0 & c & 0 & 0 \end{pmatrix}\)};
		\node [style=none] (10) at (-8.75, 2) {\((0,1)\)};
		\node [style=none] (11) at (-8.75, -1.5) {\((1,0)\)};
	\end{pgfonlayer}
	\begin{pgfonlayer}{edgelayer}
		\draw (0) to (4);
		\draw (4) to (1);
		\draw (1) to (3);
	\end{pgfonlayer}
\end{tikzpicture}
  \caption{An example of a labelled open graph (left) with corresponding \Fflow
    (right). The inputs of the open graph are square vertices, and the outputs are
    white. The labelling is written in parentheses next to the unmeasured vertices.
    The edge weights can take any values in \(\F\) with the only constraint
    being that \(a\) must be invertible. We measure the input vertex before the
    auxiliary non-output, which gives the corresponding layer decomposition.}
  \label{fig:example}
\end{figure}
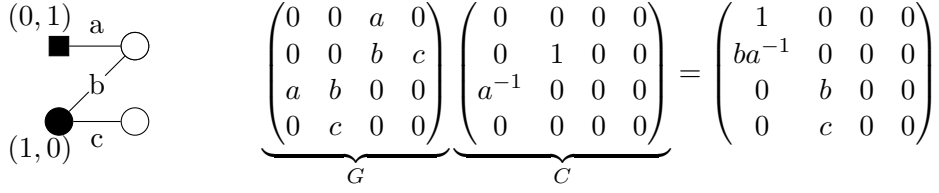

\subsection{Recovering other flow conditions}

gflow was originally formulated in terms of a partial order on the vertices to be
measured \cite{browne_generalized_2007}.
It is easy to see that the order of the layer decomposition \(\Lambda\)
induces a (non-unique) partial order on the vertices \(V\).
Given a partial order \(\prec\) on the vertices \(V\), then there is of course
a (also non-unique) ordered partition of \(V\) that agrees with \(\prec\).
Since either of these orders are only used to describe the measurement order
for the vertices of the graph, we can write the \Fflow condition in terms
which are closer to \cite{browne_generalized_2007} (stated without proof):
\begin{lemma}[Partial order \Fflow]
  \((G,I,O,\lambda)\) has a \emph{\Fflow} if and only if there exists a matrix
  \(C \in \Z_d^{V \times V}\) and a partial order \(\prec\) on \(V\) such that
  \begin{enumerate}
  \item $\forall u \in O^c$, $\lambda(u)=(C_{uu},(GC)_{uu})$;
  \item $C_{uv} =0$ whenever $u\in I$ or $v\in O$;
  \item when the columns and rows of \(G\) and \(C\) are ordered according to
    any totalisation of \(\prec\), $C$ and $GC$ are lower triangular.
  \end{enumerate}
  \label{lem:order_flow}
\end{lemma}
\vspace{-3mm}

It is straightforward from this formulation to recover the gflow condition,
since the parity conditions in the original formulation correspond to linear
equations over \(\Z_2\) (also stated without proof):
\begin{proposition}
  An open \(\Z_2\)-graph \((G,I,O,\lambda)\) has a gflow if and only if it has a
  \(\Z_2\)-flow.
\end{proposition}

Although the semantics are subtly different, we also can also recover CV-flow
\cite{booth_flow_2021} as a special case of our definition: 
\begin{proposition}
  An open \(\R\)-graph \((G,I,O,\lambda)\) such that \(\lambda(O^\mathsf{c}) =
  \{(0,1)\}\) has a CV-flow if and only if it has an \(\R\)-flow.
\end{proposition}

\subsection{The converse result}

It has been shown in the qubit case that any measurement pattern that is
robustly deterministic is such that the underlying open graph has a gflow
\cite{browne_generalized_2007}. We generalise this result to
the case of qudits:
\begin{theorem}
  If \((G,I,O,\lambda,\vb{x,z})\) is a robustly deterministic
  MBQC on \(\mathbb{Z}_d\), then the underlying open \(\F\)-graph
  \((G,I,O,\lambda)\) has an \Fflow \((C,\Lambda)\) such that
  \(\vb{x} = \vb{x}_C\) and \(\vb{z} = \vb{z}_C\).
  \label{thm:converse}
\end{theorem}

Our proof of theorem \ref{thm:converse} relies crucially on the following lemma,
which has a technical proof left to appendix \ref{app:local_tomography}:
\begin{restatable}{lemma}{LemLocalTomographyMultiple}
  Let \((G,I,O,\lambda)\) be an open graph, \(\ket{\phi},\ket{\phi'} \in
  (\mathcal{H}^{\otimes V})_1\) and \(R \subseteq V\).
  For any \(\vb{M}:R \to U(\mathcal{H})\) such that, for each \(u \in V\),
  \(\vb{M}(u) \in \mathcal{M}(\lambda(u))\), and \(\vec{m} \in \F^R\), put
  \(\ket{\vec{m}:\vb{M}} = \bigotimes_{r \in R} \ket{m_r : \vb{M}(r)}\).
  If, for every such \(\vec{m}\) and \(\vb{M}\), we have
  \begin{equation}
    \bra{\vec{m}:\vb{M}} \ket{\phi} \simeq \bra{\vec{m}:\vb{M}} \ket{\phi'} \quad
    \text{and} \quad \norm{\bra{\vec{m}:\vb{M}} \ket{\phi}} = \frac{1}{\sqrt{d^{|R|}}} =
    \norm{\bra{\vec{m}:\vb{M}} \ket{\phi'}},
  \end{equation}
  then there is a subset \(L \subseteq R\), \(\vec{x},\vec{y} \in \F^L\) and
  \(\ket{\psi} \in \mathcal{H}^{\otimes V \setminus L}\) such that
  \begin{equation}
    \ket{\phi} \simeq \ket{\psi} \bigotimes_{n \in L} \ket{x_n : Q_n} \quad \text{and} \quad \ket{\phi'} \simeq \ket{\psi} \bigotimes_{n \in L} \ket{y_n : Q_n}.
  \end{equation}
  \label{lem:local_tomography_multiple}
\end{restatable}

In the original gflow article \cite{browne_generalized_2007}, this lemma was not
taken into account in full generality, the proof of the converse result is
however fixed in \cite{perdrix_one_nodate} for the qubit case.

\begin{proof}[Proof of theorem \ref{thm:converse}]
  Let \(\prec\) be the order on \(O^\mathsf{c}\), and consider the last
  measurement made according to some totalisation of \(\prec\).
  Suppose it is made at vertex \(u\).
  Let \(\vb{M} : O^\mathsf{c} \to U(\mathcal{H})\) be such that \(\vb{M}(v) \in
  \mathcal{M}(\lambda(v))\) for all \(v \in O^\mathsf{c}\).
  Performing the measurement with outcome \(m\), there is a corresponding
  correction \(X_{\vb{x}(u)}^m Z_{\vb{z}(u)}^m\) that acts only on outputs, and
  which induces the branch map:
  \begin{align}
    \ket{G(\phi)} \longmapsto
    &X_{\vb{x}(u)}^m Z_{\vb{z}(u)}^m \left( \bra{m:\vb{M}(u)} \bigotimes_{v \in O^\mathsf{c} \setminus \{u\}} \bra{0:\vb{M}(v)} \right) \ket{G(\phi)} \\
    &= X_{\vb{x}(u)}^m Z_{\vb{z}(u)}^m \left( \bigotimes_{v \in O^\mathsf{c}} \bra{0:\vb{M}(v)} \right) Q_u^m \ket{G(\phi)}, \\
    &= \left( \bigotimes_{v \in O^\mathsf{c}} \bra{0:\vb{M}(v)} \right) X_{\vb{x}(u)}^m Z_{\vb{z}(u)}^m Q_u^m \ket{G(\phi)}\\
    &= \bigotimes_{v \in O^\mathsf{c}} \ip*{0:\vb{M}(v)}{G(\phi)}.
  \end{align}

  By the uniformity condition, this equation is true for any choice of
  measurements in \(\prod_{v \in O^\mathsf{c}} \mathcal{M}(\lambda(v))\).
  In particular, for any \(\vb{M}(v) \in \mathcal{M}(\lambda(v))\), by
  proposition \ref{prop:measurement_space_unitary} we have \(Q_v^{-m} \vb{M}(v)
  Q_v^m \in \mathcal{M}(\lambda(v))\), and
  \begin{equation}
    \bra{0:Q_v^{-m} \vb{M}(v) Q_v^m} = \bra{0:\vb{M}(v)} Q_v^m = \bra{m:\vb{M}(v)}.
  \end{equation}
  It follows that for any choice of measurements \(\vb{M}\) and any \(\vec{m}
  \in \Z_d^{O^\mathsf{c}}\),
  \begin{equation}
    \bra{\vec{m}:\vb{M}} X_{\vb{x}(u)}^m Z_{\vb{z}(u)}^m Q_u^m \ket{G(\phi)}
    = \ip*{\vec{m}:\vb{M}}{G(\phi)},
  \end{equation}
  so by lemma \ref{lem:local_tomography_multiple}, there is a subset \(L
  \subseteq O^\mathsf{c}\), vectors \(\vec{x},\vec{y} \in \Z_d^{|L|}\) and a
  state \(\ket{\psi} \in \mathcal{H}^{\otimes V \setminus L}\) such that
  \begin{equation}
    X_{\vb{x}(u)}^m Z_{\vb{z}(u)}^m Q_u^m \ket{G(\phi)} \simeq \ket{\psi} \bigotimes_{n \in L} \ket{x_n : Q_n}
    \qand \ket{G(\phi)} \simeq \ket{\psi} \bigotimes_{n \in L} \ket{y_n : Q_n}.
  \end{equation}
  Then,
  \begin{align}
    X_{\vb{x}(u)}^m Z_{\vb{z}(u)}^m Q_u^m \ket{\psi} \bigotimes_{n \in L} \ket{x_n : Q_n} &\simeq X_{\vb{x}(u)}^m Z_{\vb{z}(u)}^m Q_u^m \ket{G(\phi)} \\
    &\simeq \ket{G(\phi)} \simeq \ket{\psi} \bigotimes_{n \in L} \ket{y_n : Q_n}.
  \end{align}

  If \(u \notin L\), then since the corrections only act on outputs this implies
  that
  \begin{equation}
    (X_{\vb{x}(u)}^m Z_{\vb{z}(u)}^m Q_u^m \ket{\psi}) \bigotimes_{n \in L} \ket{x_n : Q_n} \simeq
    \ket{\psi} \bigotimes_{n \in L} \ket{y_n : Q_n}, 
  \end{equation}
  so we must have \(x_n = y_n\) for all \(n \in L\).

  If \(u \in L\), 
  \begin{align}
    (X_{\vb{x}(u)}^m Z_{\vb{z}(u)}^m Q_u^m \ket{\psi}) &\bigotimes_{n \in L} \ket{x_n : Q_n} \\
    &\simeq
      (X_{\vb{x}(u)}^m Z_{\vb{z}(u)}^m \ket{\psi}) \otimes Q_u^m \ket{x_u:Q_u} \bigotimes_{n \in L \setminus \{u\}} \ket{x_n : Q_n}, \\
    &\simeq
      (X_{\vb{x}(u)}^m Z_{\vb{z}(u)}^m \ket{\psi}) \otimes \omega^{mx_u} \ket{x_u:Q_u} \bigotimes_{n \in L \setminus \{u\}} \ket{x_n : Q_n}, \\
    &\simeq
    \ket{\psi} \otimes \ket{y_u:Q_u} \bigotimes_{n \in L \setminus \{u\}} \ket{y_n : Q_n}, 
  \end{align}
  which which also implies that \(x_n = y_n\) for all \(n \in L\).
  Then,
  \begin{equation}
    X_{\vb{x}(u)}^m Z_{\vb{z}(u)}^m Q_u^m \ket{G(\phi)} \simeq \ket{G(\phi)},
  \end{equation}
  and \(X_{\vb{x}(u)} Z_{\vb{z}(u)} Q_u\) stabilises the graph state for any \(m
  \in \Z_d\), up to a phase \(e^{i\alpha}\). By proposition
  \ref{prop:open_graph_stabilisers} there is some multiset \(C_{\bullet u} \in
  \Z_d^V\) such that
  \begin{equation}
    e^{i\alpha} X_{\vb{x}(u)} Z_{\vb{z}(u)} Q_u
    = \omega^{C_{\bullet u}^\mathsf{T} G C_{\bullet u}} X_{C_{\bullet u}} Z_{GC_{\bullet u}}
    \qand
    (C_{\bullet u})_v = 0 \qif v \in I.
  \end{equation}

  The corrections act only on outputs, so that the factor of \(X_{C_{\bullet
      u}} Z_{GC_{\bullet u}}^{-1}\) acting on \(u\) must be \(Q_u\).S
  This implies that \(X_{C_{uu}} Z_{(GC)_{uu}}^{-1} \simeq Q_u\), so that \(\lambda(u) =
  ((C_{\bullet u})_u,(GC_{\bullet u})_u\), and furthermore, that \((C_{\bullet
    u})_v = (GC_{\bullet u})_v = 0\) if \(v \notin O \cup \{u\}\) since
  \(X_{\vb{x}(u)}^m Z_{\vb{z}(u)}^m\) acts only on outputs.
  Furthermore, tensor products of Paulis form a basis of the space of linear
  operators, so that we must have
  \begin{equation}
    \vb{x}(v) \coloneqq (C_{\bullet v} - \lambda(v)_1 1_{\{v\}})
    \qand \vb{z}(v) \coloneqq ((GC)_{\bullet v} - \lambda(v)_2 1_{\{v\}}).
    \label{eq:converse_corrections_flow}
  \end{equation}
  
  Now, consider the open graph \((G,I,O \cup \{u\})\). 
  Since \((G,I,O,\lambda,\vb{x,z})\) is robustly deterministic, we can
  repeat the same procedure on the new open graph
  \begin{equation}
    (G,I,O \cup \{u\},\lambda|_{(O \cup \{u\})^\mathsf{c}},
    \vb{x}|_{(O \cup \{u\})^\mathsf{c}}, \vb{z}|_{(O \cup \{u\})^\mathsf{c}}),
  \end{equation}
  obtaining \(C_{\bullet v}\) for the
  last measured vertex \(v\) in \(O^\mathsf{c} \setminus \{u\}\).
  This procedure eventually terminates, and we end up with a column vector
  \(C_{\bullet w}\) for each \(w \in O^\mathsf{c}\).
  Let \(C \in \F^{V \times V}\) be the matrix whose \(u\)-th column is
  \(C_{\bullet u}\), or \(0\) if \(u \in O\).
  Then from the equations in lemma \ref{lem:penultimate_layer} we see that
  the pair \((C,<)\) gives an \Fflow for \((G,I,O,\lambda)\) by lemma
  \ref{lem:order_flow}.
  Furthermore, it is also clear from equations \eqref{eq:corrections_flow} and
  \eqref{eq:converse_corrections_flow} that \(\vb{x} = \vb{x}_C\) and \(\vb{z} =
  \vb{z}_C\).
\end{proof}

\section{A polynomial-time algorithm for \Fflow}
\label{sec:algorithm}
\begin{algorithmic}[1]
  \Statex \textbf{input:} A labelled open graph \((G,I,O,\lambda)\)
  \Statex\textbf{output:} An \Fflow \((C,\Lambda)\) or \textbf{fail}
  \vspace{5mm}
  \Procedure{F-Flow}{$G, I, O, \lambda$}
  \State find \(L \coloneqq \{u \in V \mid (\forall v \in V): G_{uv} = 0\}\)
  \Comment{Isolated vertices}
    \State \(\mathrm{layer}(0) \coloneqq O \cup L\)
    \State \(C \coloneqq 0_{|V|\times|V|}\)
    \State  \textbf{return} \Call{F-Flow-aux}{\(G, I, O \cup L,
      \lambda|_{O^\mathsf{c} \setminus L},
      C\), \textrm{layer}, 1}
  \EndProcedure
  \Statex
  \Procedure{F-Flow-aux}{\(G, I, O, \lambda, C\), \textrm{layer}, $k$}
    \State \(L \coloneqq \varnothing\) \Comment{Vertices which we are correcting in
    this layer}
    \ForAll{\(v \in O^\mathsf{c}\)}
      \State \((a,b) \coloneqq \lambda(v)\)
      \State solve in \(\F\): \(G[O^\mathsf{c},O \setminus I] \vec{c} =
      b1_{\{v\}} - aG[O^\mathsf{c},\{v\}]\)
      \If{there is a solution \(\vec{c}\)}
        \State \(L \coloneqq L \cup \{v\}\) \Comment{Assign \(v\) to the current
          layer}
        \State \(C[O \setminus I,\{v\}] \coloneqq \vec{c}\) \Comment{The
          corrections for vertex \(v\)}
        \State \(C[\{v\},\{v\}] \coloneqq a\)
      \EndIf
    \EndFor
    \If{\(L = \varnothing\)} \Comment{If we cannot correct for additional
    vertices, either:}
      \If{\(O = V\)} 
        \State \textbf{return} \((C,\mathrm{layer})\) \Comment{we have found an
          \Fflow; or,}
      \Else 
        \State \textbf{fail} \Comment{there is no \Fflow.}
      \EndIf
    \Else
      \State \(\mathrm{layer}(k) \coloneqq L\) 
      \State  \textbf{return} \Call{F-Flow-aux}{\(G, I, O \cup
        L,\lambda|_{O^\mathsf{c} \setminus L}, C\),
        \textrm{layer}, $k+1$}
    \EndIf
  \EndProcedure
\end{algorithmic}

\subsection{Correctness}

This algorithm is strongly inspired by the analogous algorithm for finding
gflows for \(\Z_2\)-graphs \cite{mhalla_finding_2008}, and its generalisations
to multiple measurement planes \cite{backens_there_2021} and to finding CV-flows
\cite{booth_flow_2021}.

\begin{theorem}
  \((G,I,O,\lambda)\) has an \Fflow if and only if the algorithm above
  returns a valid \Fflow.
  \label{thm:algo_correctness}
\end{theorem}
\begin{proof}
  It is clear the algorithm terminates, since at each call to
  \textsc{Z-Flow-Aux}, the algorithm either passes vertices from \(V \setminus O\)
  to \(O\), returns an \Fflow, or fails. Since \(V\) is finite, there are a
  finite number of recursions after which the algorithm either returns an \Fflow
  or fails.

  \textit{(Outputs a valid \Fflow)} Suppose the algorithm terminates with a pair
  $(H,C)$. We need to show this defines a valid $\mathbb F_q$-flow. Consider the
  function \textsc{Z-Flow-Aux} at a given call, and let \(H' \coloneqq
  G[O^\mathsf{c},O \setminus I]\) and \(h \coloneqq G[O^\mathsf{c},{v}]\).

  The output columns of $C$ are $0$, and since the solution vector $x$ from line
  11 never contains an input, the input rows of $C$ are also $0$. Hence
  condition (i) is satisfied.

  Similarly, the solution vector $x$ only has rows labelled by vertices $v \geq
  u$, so $C$ is lower triangular by construction. If the linear equation in line
  11 is satisfied, then the entries above $(HC)_{uu}$ in the $u$-th column of
  $HC$ will be $0$. Hence (ii) is satisfied. Indeed, for any $u>v$, $(HC)_{vu} =
  \sum_wH_{vw}C_{wu} = \sum_{w<u} H_{vw}C_{wu} + H_{vu}C_{uu} + \sum_{w>u}
  H_{vw}C_{wu} = h_wa + \sum_{w>u} H'_{vw}x_w$. Since $C_{uu}=a$, $H_{vu} =
  h_v$, $\forall w<u, C_{wu} =0$, and $\forall w>u, H_{vw} = H'_{vw}$ and
  $C_{wu} = x_w$. As a consequence $(HC)_{vu} = (ah+H'x)_v = (ah + b\{u\}
  -ah)_v$

  Finally, for a non-output $u$, $C_{uu} = a$. As a consequence $C$ is an
  $\mathbb F_q$-flow for
  $H$.

  \textit{(Outputs a valid layer decomposition)} Let \(L\) be as in line 21 of
  the algorithm for some call \(k\) to \textsc{Z-Flow-Aux}. It is clear that the
  equation line 10 doesn't depend on the vertices in \(L\) which appear before
  or after \(\{v\}\) and therefore \(L\) is independent of the order in which
  the elements of \(L\) are found. As a result, the output of the algorithm is
  invariant any permutation of the vertices in \(L\). Since this corresponds
  tautologically to a permutation of the layer \(V_k\) output by the algorithm,
  and every permutation that preserves the partition can be written as a product
  of such permutations, the \(\Z_d\)-flow found by the graph is invariant
  under permutations that preserve the layers (whenever the algorithm succeeds).

  \textit{(Outputs a \Fflow whenever there is one)} Assume the algorithm fails,
  that is, for some call to \textsc{Z-Flow-Aux}, line 10 has no solution for any
  remaining unfinished vertices. Let \(\overline{O}\) be the third parameter at
  that function call, and further assume that \(D\) is an \Fflow for
  \((G,I,O,\lambda)\).

  Let \(\overline{D}\) be the matrix obtained by replacing the columns in \(D\)
  corresponding to \(\overline{O}\) with zeros and permuted such that the
  columns \(\overline{O}\) appear last. Then, \(\overline{D}\) is an \Fflow for
  \((G,I,\overline{O},\lambda|_{\overline{O}^\mathsf{c}})\). Let \(v \in
  \overline{O}^\mathsf{c}\) be the last column before
  \(\overline{O}^\mathsf{c}\), and put \(c \coloneqq \overline{D}[\overline{O}
  \setminus I, \{v\}]\). Then,
  \begin{align}
    (G[\overline{O}^\mathsf{c}, \overline{O} \setminus I] c)_u
    &= \sum_{j \in \overline{O} \setminus I} G_{uj}c_j
      = \sum_{j \in \overline{O} \setminus I} G_{uj} \overline{D}_{jv}
      = \sum_{j \in \overline{O}} G_{uj} \overline{D}_{jv} = \sum_{j \in \overline{O}} G_{uj} D_{jv} \\
    &= \sum_{j \in V} G_{uj} D_{jv} - \sum_{j \in \overline{O}^\mathsf{c}} G_{uj} D_{jv}
      = (GD)_{uv} - \sum_{j \in \overline{O}^\mathsf{c}} G_{uj} D_{jv}\\
    &= b \delta_{u,v} - G_{uv}D_{vv} = b \delta_{u,v} - aG_{uv}.
  \end{align}
  As a result, we see that \(c\) verifies the equation of line 10, which
  contradicts the failure of the algorithm. It follows that \((G,I,O,\lambda)\)
  cannot have an \Fflow if the algorithm fails. By contrapositive, if
  \((G,I,O,\lambda)\) has an \Fflow, the algorithm succeeds.
\end{proof}

The core of the algorithm is the loop line 8. Letting \(n = \abs{V}\), \(\ell =
\abs{O}\) and \(\ell' = \abs{O \setminus I}\) at a given call to
\textsc{Z-Flow-Aux}, note that \(\ell' \leqslant \ell \leqslant n\). The loop
amounts to solving \(n-\ell\) systems of \(n-\ell\) equations in \(\ell'\)
variables. Let \(x_v\) be the right hand side of equation line 10. Solving the
system can be done by transforming the matrix
\begin{equation}
  \big[ G[O^\mathsf{c},O \setminus I] \mid x_{v_1} \mid \cdots \mid x_{v_{n-\ell}} \big]
\end{equation}
to upper echelon form. This can by done in time \(O(n^3)\) by Gaussian
elimination, and backsubstituting to find the corresponding \(\vec{c}_j\) to
each \(x_j\) takes time \(O(n^2)\) or for all solutions \(O(n^3)\) since there
are at most \(n\) backsubstitutions to perform. Finally, since each call to
\textsc{Z-Flow-Aux} either eliminates a vertex or terminates, the algorithm
recurses at most \(n\) times. The total complexity is therefore \(O(n^4)\).

Note this procedure can also be adapted to find an \Fflow for \textit{any}
labelling, rather than one fixed in advance. First, note that, for the existence
of an \Fflow, it suffices to choose measurement planes up to a scalar factor.
That is, \((C,\Lambda)\) is an \Fflow for \((G,I,O,\lambda)\) with $\lambda(u) =
(a,b)$ if and only if it is an \Fflow for \((G,I,O,\lambda')\) where
$\lambda'(u) = (ka,kb)$. Hence we can solve for measurement planes at the same
time as $C$ by either fixing $a = 1$ and solving for $b$ in the equation line 10
of the algorithm, or for non-inputs, fixing $b = 1$ and solving for
$a$.

\subsection{Depth optimality}

Our proofs follow the structure of \cite{mhalla_finding_2008}, which introduced
the idea of optimising gflows starting from the last layer and working back.
The idea is to find corrections for as many measured vertices as possible at
the part of the MBQC when there are the most constraints on possible
corrections: when the only vertices left unmeasured are the outputs.
This motivates the following definition which allows us to conveniently
manipulate layer decompositions ``from the back'':
\begin{definition}
  Let \((C,\Lambda)\) be an \Fflow for an open graph \((G,I,O,\lambda)\).
  Then, the \emph{depth} of \((C,\Lambda)\) is \(\abs{\Lambda}-1\).
  Furthermore, we define an \(\N\)-indexing of the elements of \(\Lambda\) by:
  \begin{equation}
    \Lambda_k \coloneqq \max( \Lambda \setminus \{\Lambda_n \mid n < k\}),
  \end{equation}
  where we note that \(\Lambda_m < \Lambda_n\) as elements of \(\Lambda\) if and
  only if \(n > m\), and \(\Lambda_k \neq \emptyset\) if and only if \(k\) is less
  than or equal to the depth of \((C,\Lambda)\).
\end{definition}

This definition of the depth of an \Fflow corresponds to the intuitive
interpretation: all measurements (and corresponding corrections) within a layer
can be made concurrently, therefore there is an implementation that runs the
MBQC in \(\abs{\Lambda}-1\) rounds of measurements (since the outputs are
not measured).

Now, we can use this definition to compare the depths of different \Fflows:
\begin{definition}
  Let \((C,\Lambda)\) and \((D,\Phi)\) be \Fflows for an open graph
  \((G,I,O,\lambda)\). Then \((C,\Lambda)\) is \emph{more delayed} than
  \((D,\Phi)\) if for each \(k\),
  \begin{equation}
    | \bigcup_{n=0}^k \Lambda_n | \geqslant | \bigcup_{n=0}^k \Phi_n |.
  \end{equation}
  and this inequality is strict for at least one \(k\). It is \emph{maximally
    delayed} if there is no layer decomposition which is more delayed.
\end{definition}

Then, we can give a complete characterisation of the layer decompositions of
maximally delayed \Fflows, which turn out to be uniquely defined:
\begin{proposition}
  If \((C,\Lambda)\) is a maximally delayed \Fflow for an open graph
  \((G,I,O,\lambda)\), then \(\Lambda_0 = O \cup \{u \in V \mid (\forall v \in V):
  G_{uv} = 0\}\) and for \(k > 0\),
  \begin{equation}
    \Lambda_k = \Bigg\{ u \in (O \bigcup_{1<n<k} \Lambda_n)^\mathsf{c} \mid \exists c \in \F^V \quad \text{s.t.} \quad
    \begin{aligned}
      &(c_u, (Gc)_u) = \lambda(u) \\
      &\forall v \notin (O \bigcup_{1<n<k} \Lambda_n) \cup \{u\}, c_v = (Gc)_v = 0
    \end{aligned}
    \Bigg\}.
  \end{equation}
  In particular, if \((C,\Lambda)\) and \((D,\Phi)\) are maximally delayed
  \Fflows for the same open graph, then \(\Lambda = \Phi\).
  \label{prop:maximal_delay_unique}
\end{proposition}

To make this proof we need the following three lemmas, which have somewhat
cumbersome proofs that we have chosen to leave to appendix \ref{app:maximal_delay_lemmas}:
\begin{restatable}{lemma}{LemFinalLayer}
  If \((C,\Lambda)\) is a maximally delayed \Fflow for an open graph
  \((G,I,O,\lambda)\), then \(\Lambda_0 = O \cup \{u \in V \mid (\forall v \in V):
  G_{uv} = 0\}\), i.e. the union of the outputs and isolated vertices of
  \((G,I,O,\lambda)\).
  \label{lem:final_layer}
\end{restatable}

\begin{restatable}{lemma}{LemPenultimateLayer}
  If \((C,\Lambda)\) is maximally delayed for \((G,I,O,\lambda)\), then
  \begin{equation}
    \Lambda_1 = \Bigg\{ u \in O^\mathsf{c} \,\mid\, \exists c \in \Z_d^{|V|} \quad \text{s.t.} \quad
    \begin{aligned}
      &(c_u, (Gc)_u) = \lambda(u) \\
      &\forall v \notin O \cup \{u\}, c_v = (Gc)_v = 0
    \end{aligned}
    \Bigg\}.
    \label{eq:penultimate_layer}
  \end{equation}
  \label{lem:penultimate_layer}
\end{restatable}

\begin{restatable}{lemma}{LemRecursiveDecomposition}
  If \((C,\Lambda)\) is a maximally delayed \Fflow of \((G,I,O,\lambda)\), \((D,
  \Phi)\) is a maximally delayed \Fflow of \((G,I,O \cup \Lambda_1,\lambda|_{(O
    \cup \Lambda_1)^\mathsf{c}})\), where
  \begin{itemize}
  \item \(D\) is the matrix obtained by replacing the columns of \(C\)
    corresponding to \(\Lambda_1\) with zeros;
  \item \(\Phi\) is given by
    \begin{equation}
      \Phi_k \coloneqq
      \begin{cases}
        \Lambda_1 \cup O \quad \text{if} \quad k = 0; \\
        \Lambda_{k+1} \quad \text{otherwise}.
      \end{cases}
    \end{equation}
  \end{itemize}
  \label{lem:recursive_decomposition}
\end{restatable}

\begin{proof}[Proof of proposition \ref{prop:maximal_delay_unique}]
  \(\Lambda_0\) must take the form given in lemma \ref{lem:final_layer}.
  Then, a recursive application of lemma \ref{lem:recursive_decomposition} and
  lemma \ref{lem:penultimate_layer} shows that the layer decomposition of a
  maximally delayed \Fflow is uniquely defined.

  Since the open graph obtained from lemma \ref{lem:recursive_decomposition} and
  used to calculate \(\Lambda_k\) with lemma \ref{lem:penultimate_layer} is
  \((G,I,O \bigcup_{1<n<k} \Lambda_n, \lambda|_{(O \bigcup_{1<n<k}
    \Lambda_n)^\mathsf{c}})\), it is clear that \(\Lambda_k\) must take the form
  claimed.
\end{proof}

This can be understood from the following principle. If a correction exists for
the measurement of a vertex that acts only on outputs, then this correction can
be performed at any point during the MBQC, since the outputs are never measured
and therefore always available for corrections. As a result, we can delay this
measurement as much as possible, to the penultimate layer, to give ourselves as
much flexibility as possible in corrections for previous layers. As a result, we
can put all vertices whose corrections act only on outputs in the final layer.
Any vertices which do not verify this property must be in a layer which precedes
the penultimate layer, since at the time they are measured there must be
non-output vertices which are left unmeasured.
Then, it suffices to show that there is a minimal depth \Fflow that is maximally
delayed to obtain:
\begin{proposition}
  A maximally delayed \Fflow for an open graph \((G,I,O,\lambda)\) has minimal
  depth.
  \label{prop:maximal_optimal_depth}
\end{proposition}
\begin{proof}
  First, note that if \((C,\Lambda)\) is more delayed than \((D,\Phi)\), then in
  particular,
  \begin{equation}
    \abs{V} = | \bigcup_{n=0}^\abs{\Lambda} \Lambda_n |
    \geqslant | \bigcup_{n=0}^\abs{\Lambda} \Phi_n |,
  \end{equation}
  so that \(\abs{\Lambda} \leqslant \abs{\Phi}\).
  Assume now that \((D,\Phi)\) has minimal depth, then any \Fflow that is more
  delayed has the same depth.
  It follows that either \((D,\Phi)\) is maximally delayed and has minimal
  depth, or there is a maximally delayed \Fflow that is more delayed than
  \((D,\Phi)\) thus has the same depth.
  But by proposition \ref{prop:maximal_delay_unique}, every maximally delayed
  \Fflow has the same layer decomposition for a given open graph, so that every
  maximally delayed \Fflow has minimal depth.
\end{proof}

Note however that a minimal depth decomposition is not necessarily maximally
delayed. For example, we can always measure the entirety of the inputs first
without changing the depth, but this measurement order is not always maximally
delayed since this allows us to move inputs into earlier layers.
Since the algorithm is constructed such that it finds a maximally delayed
\Fflow:
\begin{theorem}
  The algorithm outputs an \Fflow of minimal depth.
  \label{thm:algo_optimal}
\end{theorem}
\begin{proof}
  Assume that the algorithm succeeds with output \((D,\Phi)\).
  We show that this output \Fflow \((D,\Phi)\) is maximally delayed.
  Firstly, we show that the output flow has \(\Phi_1 = \Lambda_1\) from lemma
  \ref{lem:penultimate_layer}.
  We know that \(\Phi_1 \subseteq \Lambda_1\) and it is clear from the
  definition of the algorithm that
  \begin{align}
    \Phi_1 &= \Big\{u \in O^c \hspace{1mm}\mid\hspace{1mm}
          \exists \vec{c} \in \Z_d^{|O \setminus I|} \quad \text{s.t.} \quad
          G[O^\mathsf{c},O \setminus I] \vec{c} = b1_{\{v\}} - aG[O^\mathsf{c},\{v\}] \Big\}
  \end{align}

  Let \(u \in \Lambda_1\), that is there is some vector \(\vec{c} \in
  \Z_d^{|V|}\) such that
  \begin{equation}
    \begin{cases}
      (c_u, (Gc)_u) = \lambda(u) \\
      \forall v \notin O \cup \{u\}, c_v = (Gc)_v = 0
    \end{cases}
  \end{equation}
  Then, for any \(v \in O^\mathsf{c}\),
  \begin{align}
    (G[O^\mathsf{c},V] c)_v &= \sum_{j \in V} G_{vj}c_j = \sum_{j \in O \cup \{u\}} G_{vj}c_j \\ 
    &= \sum_{j \in O} G_{vj}c_j + aG_{vu} = (Gc)_v + aG_{vu} = b \delta_{vu} + aG_{vu}
  \end{align}
  from which we see that \(u \in \Phi_1\) whence \(\Phi_1 = \Lambda_1\).

  Now, \(\Phi_2\) is calculated in the next call to \textsc{Z-Flow-Aux} where
  the open graph passed as argument is \((G,I,O \cup \Phi_1,\lambda|_{(O \cup
    \Phi_1)})\).
  Using the same argument as for \(\Phi_1\), \(\Phi_2\) must match the layer
  \(\Lambda_2\) obtained by applying lemma \ref{lem:penultimate_layer} to the
  \Fflow resulting from \ref{lem:recursive_decomposition}.
  
  Then, using the same recursion as in the proof of proposition
  \ref{prop:maximal_delay_unique}, we see that \((D,\Phi)\) is maximally delayed.
  It follows from proposition \ref{prop:maximal_optimal_depth} that the \Fflow
  output by the algorithm has optimal depth.
\end{proof}

\section*{Conclusion and future work}

We have defined a flow condition suitable for qudit MBQC, and shown that it is
sufficient to obtain deterministic MBQCs.
We leave two major open questions for future work: firstly, we have only
considered the case of fields of prime cardinality. Most of our results
generalise straightforwardly to the case of power-of-prime fields, with the
exception of theorem \ref{thm:converse} (or more specifically, lemma
\ref{lem:local_tomography}). Secondly, an important tool for studying MBQC is
ancilla-less circuit extraction. Work in this direction was started by the first
and third authors in \cite{booth_flow_2021} where an algorithm was found for
measurement patterns where all the measurements belong to \(\mathcal{M}(0,1)\)
(the measurement space of \(Z\)), but no extraction algorithm is known for all
measurement spaces.

\paragraph{Acknowledgements.} R.\,I.\,B., D.\,M. and S.\,P. were supported by the
ANR VanQuTe project (ANR-17-CE24-0035).

\printbibliography

\appendix
\section{Proof of lemma \ref{lem:local_tomography_multiple}}
\label{app:local_tomography}
\begin{lemma}
  Let \(\ket{\phi}\) be a state of a register \(V\) of qudits,
  \(Q\) a Pauli operator and fix some \(v \in V\). If for every measurement \(M
  \in \mathcal{M}(Q_v)\) of the qudit \(v\) and every \(m \in \mathbb{F}\), we
  have
  \begin{equation}
    \| \langle m : M \mid \phi \rangle \| = \frac{1}{\sqrt{d}}
    \label{eq:tomographyMaximallyEntangled}
  \end{equation}
  then \(\ket{\phi}\) has a Schmidt decomposition of the form
  \begin{equation}
      \ket{\phi} = \sum_{x \in \F} c_x \ket{x:Q}\otimes \ket{\psi_x},
  \end{equation}
  where \(\ket{x : Q}\) is an eigenvector of \(Q\) associated with eigenvalue
  \(\omega^x\), and we take the coefficients \(c_x\) to be real and
  non-negative.
  \label{lem:local_tomography}
\end{lemma}

\begin{proof}
  Pick some \(M \in \mathcal{M}(Q)\), we can write
  \begin{equation}
    \ket{\phi} = \sum_{m \in \F} \ket{m : M} \ket{\phi_m}
    \qq{where} \ket{\phi_m} \coloneqq \ip{m:M}{\phi}
  \end{equation}
  and \(\norm{\ip{m:M}{\phi}} = \frac{1}{\sqrt{d}}\).\newline\newline
  Letting \(\{\ket{\psi_m}\}\) be the collection of vectors obtained by
  orthonormalising \(\{\ket{\phi_m}\}\), we can expand \(\ket{\phi}\)
  in this basis:
  \begin{equation}
    \ket{\phi} = \frac{1}{\sqrt{d}} \sum_{m,n \in \F} \Psi_{mn} \ket{m:M} \ket{\psi_n},
    \quad\text{and for any}\:m \in \F, \quad \|\Psi_{m\bullet}\|^2 = 1 
  \end{equation}
  where \(\Psi\) is therefore a \(d \times p\) matrix such that \(p\) is
  the dimension of the subspace of \(\mathcal{H}\) generated by the
  \(\ket{\phi_m}\) and we denote \(\Psi_{m\bullet}\) the $m^{\textrm{th}}$ line vector of $\Psi$.\newline

  We know that, for every rotation $U$ in $SU(d)$ preserving \(Q\) and
  every $M$ in $\mathcal{M}(Q_v)$, $UMU^\dagger$ is also in $\mathcal{M}(Q_v)$.
  The group of all such rotations acts on \(\Psi\) from the left via the Hilbert
  space representation, and this action is generated by the rotations of the form
  \(V_M^{-1}R_{k,l}(\xi)V_M\), where \(V_M\) is the $d$-dimensional discrete
  Fourier transform matrix\footnote{Explicitely, \(V_M\) is given by
    \(\mel{m:M}{V_M}{n:M} = \phi(mn)\).} in the eigenbasis of \(M\), and \(R_{k,l}(\xi)\) is
  the diagonal matrix given by \(k \in \F\), \(l \in \F^*\) and $\xi \in \R$, by
  \begin{equation}
    R_{k,l}(\xi)_{mm} \coloneqq
    \begin{cases}
      e^{-i\xi} \qif m = k; \\
      e^{i\xi} \qif m = k+l; \\
      1 \qq{otherwise.}
    \end{cases}
  \end{equation}
  According to Eq.~(\ref{eq:tomographyMaximallyEntangled}), applying a rotation preserving $Q$ to $v$ preserves the outcomes' probabilities.
  As such, we deduce that the action of rotations \(V_M^{-1}R_{k,l}(\xi)V_M\) on matrix $\Psi$ will preserve
  the norm of its line vectors. Namely, for every \(k \in \F\), \(l \in \F^*\) and $\xi \in \R$,
  \begin{equation}
    \| \Psi_{m\bullet} \|^2 = \left\| \left(D_{k,l,\xi}\Psi\right)_{m\bullet} \right\|^2 \quad \text{ where, } \quad D_{k,l,\xi} := V_M^{-1}R_{k,l}(\xi)V_M.
    \label{eq:normEquality}
  \end{equation}
  Below, we explicit the right side of this equality to find which $\Psi$ satisfy Eq.~(\ref{eq:normEquality}). First, we
  compute the transformed matrix' line vectors:
  \begin{align}
    (D_{k,l,\xi}\Psi)_{m \bullet}
    &= \Psi_{m\bullet} + \frac{1}{d} \sum_{\alpha \in \F} \Psi_{\alpha\bullet } \left( \phi(k(m-\alpha)) (e^{-i\xi} - 1) +
      \phi((k+l)(m-\alpha)) (e^{i\xi} - 1) \right) \\
    &= \Psi_{m\bullet} + P_m^{k,l,1} \sin\xi + P_m^{k,l,2} (\cos\xi - 1)
  \end{align}
  where
  \begin{align}
    P_m^{k,l,1} &\coloneqq
                  -\frac{2}{d} \sum_{\alpha \in \F} \Psi_{\alpha \bullet } \omega^{(k+\frac{l}{2})(m-\alpha)} \sin(\frac{\pi l}{d}(m-\alpha))\text{ and }\\
    P_m^{k,l,2} &\coloneqq
                  \frac{2}{d} \sum_{\alpha \in \F} \Psi_{\alpha \bullet } \omega^{(k+\frac{l}{2})(m-\alpha)} \cos(\frac{\pi l}{d}(m-\alpha)).
  \end{align}

  We rewrite Eq.~(\ref{eq:normEquality}) as,
  \begin{align}
    \| \Psi_{m\bullet} \|^2 &=  \left\| \left(D_{k,l,\xi}\Psi\right)_{m\bullet} \right\|^2 \\
                            &=  \left\|\Psi_{m \bullet} +  P_m^{k,l,1}\sin\xi + P_m^{k,l,2} (\cos\xi - 1)\right\|^2\\
                            &= \| \Psi_{m\bullet} \|^2 + A + B\sin\xi + C\cos\xi + D\cos 2\xi + E \sin 2\xi,
  \end{align}
  from which we deduce:
  \begin{equation}
    A + B\sin\xi + C\cos\xi + D\cos 2\xi + E \sin 2\xi = 0. \label{eq:fourierDecompo}
  \end{equation}
  We specify these five alphabetic constants for our kind reader while emphasizing that only the expression of $D$ will be used thereafter: 
  \begin{subequations}
    \label{eq:dirtyEquations}
    \begin{align}
      A &:=  \frac{3}{2}\left\|P_m^{k,l,2}\right\|^2 - 2\Re\left(\Psi_{m \bullet}P_m^{k,l,2 *}\right) + \frac{1}{2}\left\|P_m^{k,l,1}\right\|^2, \\
      B &:= 2\Re\left(\Psi_{m \bullet}P_m^{k,l,1 *}\right) - 2\Re\left(P_m^{k,l,1}P_m^{k,l,2 *}\right), \\
      C &:= 2\Re\left(\Psi_{m \bullet}P_m^{k,l,2 *}\right) - 2\left\|P_m^{k,l,2}\right\|^2, \\
      D&:=\frac{1}{2}\left(\left\|P_m^{k,l,2}\right\|^2 - \left\|P_m^{k,l,1}\right\|^2\right), \\
      E&:=2\Re\left(P_m^{k,l,1}P_m^{k,l,2 *}\right),
    \end{align}
  \end{subequations}
  where $P_m^{k,l,i *}$ denotes the complex conjugate of $P_m^{k,l,i}$.\newline
  We know that \{$\cos(m\xi)$, $\sin(n\xi)\}_{m,n\in \N}$ forms an orthogonal set in the space of periodic functions of period $2\pi$ with respect to the Hermitian form $\langle f,g\rangle := \int_{-\pi}^{\pi}f^*(t)g(t)\textrm{d}t$, and as such, the five alphabetic constants of the left side of Eq.~(\ref{eq:fourierDecompo}) must be zero.\newline\newline
  We develop the two terms of the constant $D$, $\forall m, k\in\F$, $l\in\F^*$ , and obtain:
  \begin{subequations}
    \begin{align}
      \left\|P_m^{k,l,1}\right\|^2 &= \frac{4}{d^2}\sum_{\alpha,\alpha'\in\F} \Psi_{\alpha \:\bullet}^*\Psi_{\alpha' \bullet} \omega^{(k+\frac{l}{2})(\alpha - \alpha')}\cos\left(\frac{\pi l}{d}(m-\alpha)\right)\cos\left(\frac{\pi l}{d}(m - \alpha')\right),\\
      \left\|P_m^{k,l,2}\right\|^2 &= \frac{4}{d^2}\sum_{\alpha,\alpha'\in\F} \Psi_{\alpha \:\bullet }^*\Psi_{\alpha' \bullet }\omega^{(k+\frac{l}{2})(\alpha - \alpha')}\sin\left(\frac{\pi l}{d}(m-\alpha)\right)\sin\left(\frac{\pi l}{d}(m - \alpha')\right).
    \end{align}
  \end{subequations}
  Using the addition formulas of trigonometry, we deduce, 
  \begin{equation}
    D = \frac{2}{d^2}\sum_{\alpha,\alpha'\in\F} \Psi_{\alpha\: \bullet }^* \Psi_{\alpha' \bullet }\omega^{(k+\frac{l}{2})(\alpha - \alpha')}\cos\left(\frac{\pi l}{d}(2m - \alpha - \alpha')\right) = 0.
  \end{equation}
  We introduce the following change of variables \(2\beta \coloneqq \alpha + \alpha'\) and \(2\beta' \coloneqq \alpha - \alpha'\), such that we obtain,
  \begin{equation}
    \forall k,n \in \F\text{ and } l\in\F^*,\quad\quad \sum_{\beta,\beta'\in \F} \omega^{2(k+\frac{l}{2})\beta'} \Psi_{\beta + \beta' \bullet  }^*\Psi_{  \beta - \beta' \bullet}\cos\left(\frac{2\pi l}{d}(m - \beta)\right) = 0.
  \end{equation}
  Now, for any $l \in \F^*$, the square matrix given by $\Omega_{k,\beta'} := \omega^{2(k+\frac{l}{2})\beta'}$ is invertible.
  As a consequence, we deduce from the previous equation that $\forall m\in \F$ and $\forall l\in \F^*$,
  \begin{equation}
    \sum_{\beta\in \F}  \Psi_{\beta + \beta' \bullet}^*\Psi_{ \beta - \beta' \bullet}\cos\left(\frac{2\pi l}{d}(m - \beta)\right) = 0.
  \end{equation}
  Developing the cosine, we obtain
  \begin{equation}
    \cos\left(\frac{2\pi lm}{d}\right)\sum_{\beta\in \F}  \Psi_{\beta + \beta' \bullet }^*\Psi_{\beta - \beta' \bullet }\cos\left(\frac{2\pi l\beta}{d}\right) + \sin\left(\frac{2\pi lm}{d}\right)\sum_{\beta\in \F}  \Psi_{\beta + \beta' \bullet }^*\Psi_{\beta - \beta' \bullet }\sin\left(\frac{2\pi l\beta}{d}\right) = 0,
  \end{equation}
  from which we deduce, using again the argument used in Eq.~(\ref{eq:fourierDecompo}), that $\forall l\in \F^*$,
  \begin{subequations}
    \begin{align}
      \sum_{\beta\in \F}  \Psi_{\beta + \beta' \bullet }^*\Psi_{\beta - \beta' \bullet }\cos\left(\frac{2\pi l\beta}{d}\right) &= 0,\\
      \sum_{\beta\in \F}  \Psi_{\beta + \beta' \bullet }^*\Psi_{\beta - \beta' \bullet }\sin\left(\frac{2\pi l\beta}{d}\right) &= 0.
    \end{align}
  \end{subequations}
  These equations force the following conclusion: for all $\beta\in\F$, the Hermitian product of $\Psi_{\beta + \beta' \bullet }$ and $\Psi_{\beta - \beta' \bullet }$ depends only of $\beta'$, namely:
  \begin{equation}
    \Psi_{\beta + \beta' \bullet }^*\Psi_{\beta - \beta' \bullet } = r_{\beta'}. \label{eq:constantProduct}
  \end{equation}\newline
  At this point, we define a ``Fourier transform'' of our line vectors $\Psi_{m\bullet}$ as
  \begin{equation}
    \Psi^F_{\gamma\bullet} \coloneqq \frac{1}{\sqrt{d}}\sum_{m\in\F} \Psi_{m \bullet} \omega^{m\gamma}.
  \end{equation}
  This tranformation is invertible as:
  \begin{equation}
    \Psi_{m \bullet} = \frac{1}{\sqrt{d}}\sum_{\gamma\in\F} \Psi^F_{\gamma \bullet} \omega^{-m\gamma},
  \end{equation}
  so that going back to $\ket{\phi}$,
  \begin{align}
    \ket{\phi} &= \frac{1}{\sqrt{d}}\sum_{m,n\in\F} \Psi_{mn}\ket{m}\ket{\psi_n}\\
               &=\frac{1}{d}\sum_{m,n\in\F}\left(\sum_{\gamma \in\F}\Psi^F_{\gamma n}\omega^{-m\gamma}\right)\ket{m}\ket{\psi_n}\\
               &= \frac{1}{d}\sum_{m,\gamma\in\F}\omega^{-m\gamma}\ket{m}\sum_{n\in\F}\Psi^F_{\gamma n}\ket{\psi_n}\\
               &= \frac{1}{\sqrt{d}}\sum_{\gamma\in\F} \ket{-\gamma:Q}\ket{\psi^{F}_{\gamma}},
  \end{align}
  where $\ket{\psi^{F}_{\gamma}} := \sum_{n\in\F}\Psi^F_{\gamma n}\ket{\psi_n}$. 
  Making good use of Eq.\@(\ref{eq:constantProduct}), we find that for $\gamma_1,\gamma_2\in\F$
  \begin{align}
    \braket{\psi^F_{\gamma_1}}{\psi^F_{\gamma_2}} &= \sum_{n\in\F}\Psi^{F*}_{\gamma_1 n}\Psi^{F}_{\gamma_2 n'}\braket{\psi_n}{\psi_{n'}}\\
                                                  &= \sum_{n\in\F}\Psi^{F*}_{\gamma_1 n}\Psi^{F}_{\gamma_2 n}\\
                                                  &= \frac{1}{d}\sum_{n\in\F}\left(\sum_{m_1,m_2\in\F}\Psi_{m_1n}^*\Psi_{m_2n}\omega^{-m_1\gamma_1 + m_2\gamma_2}\right)\\
                                                  &=\frac{1}{d}\sum_{m_1,m_2\in\F}\left(\sum_{n\in\F}\Psi_{m_1n}^*\Psi_{m_2n}\right)\omega^{-m_1\gamma_1 + m_2\gamma_2}\\
    \textrm{according to Eq.~(\ref{eq:constantProduct}),}\quad&= \frac{1}{d}\sum_{m_1,m_2\in\F}r_{\frac{m_1-m_2}{2}}\omega^{-m_1\gamma_1 + m_2\gamma_2}\\
                                                  &= \frac{1}{d}\sum_{\alpha_1,\alpha_2\in\F}r_{\alpha_2}\omega^{-(\alpha_1 + \alpha_2)\gamma_1 + (\alpha_1 - \alpha_2)\gamma_2}\\
    \textrm{summing over $\alpha_1$, }\quad &= \sum_{\alpha_2\in\F}r_{\alpha_2}\omega^{-\alpha_2(\gamma_1 + \gamma_2)}\delta_{\gamma_1,\gamma_2}.
  \end{align}
  The family $\{\ket{\psi_\gamma^F}\}_{\gamma\in\F}$ forms an orthogonal family.
  Note that, depending on the value of the $r_\alpha$, some $\ket{\psi_\gamma^F}$ can be of norm $0$. Nevertheless, whenever the condition of Eq.~(\ref{eq:tomographyMaximallyEntangled}) is met, we have a valid Schmidt decomposition of $\ket{\phi}$ of the form
  \begin{equation}
    \frac{1}{\sqrt{d}}\sum_{\gamma\in\F} \ket{-\gamma:Q}\ket{\psi^{F}_{\gamma}}.
  \end{equation}

\end{proof}

\begin{lemma}
  Let \(\ket{\phi}, \ket{\phi'}\) be two states of a register \(V\) of qudits,
  \(x \in \F^2\) be non-zero and fix some \(n \in N\). If for every measurement
  \(M \in \mathcal{M}(x)\) of the qudit \(n\) and every \(m \in \F\), we
  have
  \begin{equation}
    \langle m : M \mid \phi \rangle \simeq \langle m : M \mid \phi' \rangle \quad
    \text{and} \quad \| \langle m : M \mid \phi \rangle \| = \frac{1}{\sqrt{d}} =
    \| \langle m : M \mid \phi' \rangle \|,
    \label{eq:tomography}
  \end{equation}
  then at least one of the following holds:
  \begin{enumerate}
  \item \(\ket{\phi} \simeq \ket{\phi'}\);
  \item \(\ket{\phi}\) and \(\ket{\phi'}\) are separable and there are \(x,y\in \F\) and \(\ket{\psi} \in \mathcal{H}_{V \setminus \{v\}}\) such
    that
    \begin{equation}
      \ket{\phi} = \ket{x : Q}_v \otimes \ket{\psi} \quad \text{and} \quad
      \ket{\phi'} \simeq \ket{y : Q}_v \otimes \ket{\psi},
    \end{equation}
    where \(\ket{x : Q}\) is the eigenvector of \(Q\) associated with eigenvalue
    \(\omega^x\).
  \end{enumerate}
  \label{lem:local_tomography}
\end{lemma}

\begin{proof}
  Assume that both \(\ket{\phi},\ket{\phi'}\) have Schmidt rank \(1\). According
  to the previous lemma, we can write both states as
  \begin{equation}
    \ket{\phi} = \ket{x:Q} \otimes \ket{\psi_x}
    \qand \ket{\phi'} = \ket{y:Q} \otimes \ket{\psi'_y},
  \end{equation}
  using Eq.~(\ref{eq:tomography}),
  \begin{equation}
    \ket{\psi_x} = \sqrt{d} \ip{0:M}{\phi}
    = e^{i\alpha} \sqrt{d} \ip{0:M}{\phi'} = e^{i\alpha} \ket{\psi_y'},
  \end{equation}
  and we are clearly in subcase (2) of the main lemma.

  Now, assuming the Schmidt rank along the partition $\{v;V\setminus\{v\}\}$ of
  both \(\ket{\phi}\) and \(\ket{\phi'}\) is greater than or equal to \(2\).
  According to the previous lemma,
  \begin{equation}
    \ket{\phi} = \sum_{x \in \F} c_x \ket{x:Q}\otimes \ket{\psi_x} \qq{and}
    \ket{\phi'} = \sum_{x \in \F} c'_x \ket{x:Q}\otimes \ket{\psi'_x}.
  \end{equation}
  Then, for any \(m,k,l \in \F\), and any \(\xi \in \mathbb{T}^d\), we have
  \begin{align}
    \ip{m:M}{\phi} &= e^{i\alpha_m} \ip{m:M}{\phi'}, \label{eq:phaseEquality} \\
    \mel{m:M}{D_{k,l,\xi}}{\phi} &= e^{i\beta(k,l,\xi,m)} \mel{m:M}{D_{k,l,\xi}}{\phi'}, \label{eq:phaseEqualityGen}
  \end{align}
  where $D_{k,l,\xi}$ is defined as in Eq.~(\ref{eq:normEquality}) and $\beta$
  is a function of the different parameters which define the rotation.
  Developing the right-hand side of the previous equation we find
  \begin{align}
    \mel{m:M}{D_{k,l,\xi}}{\phi'} &= \sum_x c_x' \mel{m:M}{D_{k,l,\xi}}{x:Q} \ket{\psi'_x} \\
                                  &= \sum_x \left[\omega^{mx}
                                    + \frac{1}{d} \sum_n \omega^{nx} \left( \omega^{k(m-n)} \left(e^{-i\xi} - 1\right)
                                    + \omega^{(k+l)(m-n)} \left(e^{i\xi} - 1\right) \right)\right] c'_x \ket{\psi'_x}.
  \end{align}
  Likewise, for the left-hand side, we have for any \(m,k,l \in \F\), and any
  \(\xi \in \mathbb{T}^d\),
  \begin{align}
    \bra{m:M} &D_{k,l,\xi} \ket{\phi} \\ &= \ip{m:M}{\phi}
                                           + \frac{1}{d} \sum_n \Big(\omega^{j(m-n)} (e^{-i\xi}-1)
                                           \omega^{(j+k)(m-n)} (e^{i\xi}-1) \Big) \ip{n:M}{\phi} \\
              &= e^{i\alpha_m} \ip{m:M}{\phi'}
                + \frac{1}{d} \sum_n \Big(\omega^{j(m-n)} (e^{-i\xi}-1)
                \omega^{(j+k)(m-n)} (e^{i\xi}-1) \Big) e^{i\alpha_n} \ip{n:M}{\phi'} \\
              &= \sum_{x} \left[ e^{i\alpha_m} \omega^{mx}
                + \frac{1}{d} \sum_n e^{i\alpha_n} \omega^{nx} \left( \omega^{k(m-n)} (e^{-i\xi}-1)
                \omega^{(k+l)(m-n)} (e^{i\xi}-1) \right) \right] c'_x \ket{\psi'_x},
  \end{align}
  where we have used Eq.~(\ref{eq:phaseEquality}) between the first two lines.
  By identifying components along the orthonormal basis elements
  \(\{\ket{\psi_x'}\}\) and removing terms where \(c_x' = 0\), we can write
  Eq.~(\ref{eq:phaseEqualityGen}) as
  \begin{equation}
    \begin{aligned}
      e^{i\beta(j,k,\xi,m)} &\Big( \omega^{mx} + \frac{1}{d} \sum_n \omega^{nx}
      \omega^{k(m-n)} \left(e^{-i\xi} - 1\right)
      + \frac{1}{d} \sum_n \omega^{nx} \omega^{(k+l)(m-n)} \left(e^{i\xi} - 1\right) \Big) \\
      =& \omega^{mx} e^{i\alpha_m} + \frac{1}{d} \sum_n e^{i\alpha_n}
      \omega^{nx} \omega^{k(m-n)} \left(e^{-i\xi} - 1\right) + \frac{1}{d}
      \sum_n e^{i\alpha_n} \omega^{nx} \omega^{(k+l)(m-n)} \left(e^{i\xi} -
        1\right).
    \end{aligned}
  \end{equation}

  Since \(\ket{\phi'}\) has Schmidt rank of at least \(2\), we can find \(y,z
  \in \F\) such that \(y \neq z\), \(c_y' \neq 0\) and \(c_z' \neq 0\). For
  the next part, let \(k=y\) and \(l=z-y\) such that the phase of $D_{k,l,\xi}$
  is applied on the two non-zero components.\newline From now on, we note
  \(\beta(\xi,m) \coloneqq \beta(y,z-y,\xi,m)\). Taking the coefficients along
  \(\ket{\psi_y'}\), we rewrite the previous equation, for any \(\xi \in
  \mathbb{T}\) and \(m \in \F\), as
  \begin{equation}
    \begin{aligned}
      e^{i\beta(\xi,m)} &\omega^{my} e^{-i\xi} \\&= \omega^{my} e^{i\alpha_m} +
      \frac{1}{d} \sum_n e^{i\alpha_n} \omega^{ny} \omega^{y(m-n)}
      \left(e^{-i\xi} - 1\right) + \frac{1}{d} \sum_n e^{i\alpha_n} \omega^{ny}
      \omega^{z(m-n)} \left(e^{i\xi} - 1\right),
    \end{aligned}
  \end{equation}
  taking the coefficients along \(\ket{\psi_z'}\) we extract a different
  equation,
  \begin{equation}
    \begin{aligned}
      e^{i\beta(\xi,m)} &\omega^{mz} e^{i\xi} \\&= \omega^{mz} e^{i\alpha_m} +
      \frac{1}{d} \sum_n e^{i\alpha_n} \omega^{nz} \omega^{y(m-n)}
      \left(e^{-i\xi} - 1\right) + \frac{1}{d} \sum_n e^{i\alpha_n} \omega^{nz}
      \omega^{z(m-n)} \left(e^{i\xi} - 1\right).
    \end{aligned}
  \end{equation}

  Finally, for any \(\xi \in \mathbb{T}\) and \(m \in \F\),
  \begin{equation}
    e^{i\beta(\xi,m)}
    =  e^{i\xi}e^{i\alpha_m}
    + \frac{1}{d} \sum_n e^{i\alpha_n} \left(1 -  e^{i\xi}\right)
    + \frac{1}{d} \sum_n e^{i\alpha_n} \omega^{(z-y)(m-n)} \left(e^{i2\xi} -  e^{i\xi}\right)
  \end{equation}
  and
  \begin{equation}
    e^{i\beta(\xi,m)}
    =  e^{-i\xi}e^{i\alpha_m}
    + \frac{1}{d} \sum_n e^{i\alpha_n} \omega^{(y-z)(m-n)} \left(e^{-i2\xi} -  e^{-i\xi}\right)
    + \frac{1}{d} \sum_n e^{i\alpha_n} \left(1 -  e^{-i\xi}\right). 
  \end{equation}
  So, the right sides of both equations are equal. However, we can use again the
  argument below Eq.~(\ref{eq:dirtyEquations}), $\{ e^{m\xi}\}_{m\in \N}$ is an
  orthogonal set in the space of periodic functions. As such, taking the terms
  in $e^{2\xi}$ and $e^{\xi}$,
  \begin{subequations}
    \begin{align}
      \sum_{n}e^{i\alpha_n}\omega^{(z-y)(m-n)} &= 0\\
      e^{i\alpha_m} -  \frac{1}{d} \sum_n e^{i\alpha_n} - \sum_{n}e^{i\alpha_n}\omega^{(z-y)(m-n)} &= 0.
    \end{align}
  \end{subequations}

  Instantly, we get, for all \(m\),
  \begin{equation}
    e^{i\alpha_m} = \frac{1}{d} \sum_n e^{i\alpha_n} \qq{and in particular} e^{i\alpha_m} = e^{i\alpha_0}.
  \end{equation}
  We note this common phase $\alpha$, and this implies by direct calculation
  that
  \begin{equation}
    c_x \ket{\psi_x} = e^{i\alpha} c'_x \ket{\psi'_x}
  \end{equation}
  Based on this result, we can conlude that we are in subcase (1) of the lemma:
  \begin{align}
    \ket{\phi} &= \sum_x c_x \ket{x:Q} \otimes \ket{\psi_x}, \\
               &= \sum_x \ket{x:Q} \otimes (c_x \ket{\psi_x}), \\
               &= \sum_x \ket{x:Q} \otimes (e^{i\alpha} c'_x \ket{\psi'_x}),\\
               &= e^{i \alpha} \sum_x c'_x \ket{x:Q} \otimes \ket{\psi'_x}, \\
               &= e^{i\alpha}\ket{\phi'},
  \end{align}
  as desired.

  We have shown that any choice of \(\ket{\psi},\ket{\psi'}\) which verify the
  conditions of equation \eqref{eq:tomography} must fall into either subcase (1)
  or (2) of the lemma, and we are done.
\end{proof}

\LemLocalTomographyMultiple*
\begin{proof}
  The proof proceeds by induction on the size of \(R\).
  The case \(|R| = 0\) is trivial, and the case \(|R|=1\) is lemma
  \ref{lem:local_tomography}.
  Assume the statement is true for some non-empty \(R\), if \(R = V\) we
  are done since the induction cannot continue.
  If this is not the case, pick \(u \in V \setminus R\).
  If
  \begin{equation}
    \begin{aligned}
      (\bra{m:\vb{M}(u)} \otimes \bra{\vec{m}:\vb{M}}_R) \ket{\phi} \simeq
      (\bra{m:\vb{M}(u)} \otimes \bra{\vec{m}:\vb{M}}_R) \ket{\phi'} \\
      \qand \norm{\sqrt{d^{|R|}} (\bra{m:\vb{M}(u)}_u
        \otimes \bra{\vec{m}:\vb{M}}_R) \ket{\phi}} = \frac{1}{\sqrt{d}}
    \end{aligned}
  \end{equation}
  hold for all \(m \in \mathbb{F}\), then by
  lemma \ref{lem:local_tomography} we have one of the following cases:
  \begin{enumerate}
  \item \(\braket{\vec{m}:\vb{M}}{\phi} \simeq \braket{\vec{m}:\vb{M}}{\phi'}\) and
    \(\norm{\braket{\vec{m}:\vb{M}}{\phi}}  = \frac{1}{\sqrt{d^{|R|}}}\) for any
    \(\vec{m} \in \F^R\) so that by the induction hypothesis we are
    done.
  \item For each \(\vec{m} \in \F^R\), there are \(x,y \in \F\) and
    \(\ket{\psi_{\vec{m}:\vb{M}}} \in \mathcal{H}^{\otimes V \setminus \{u\}}\) such
      that \(\braket{\vec{m}:\vb{M}}{\phi} \simeq \ket{x : Q_u}_u \otimes
      \ket{\psi_{\vec{m}:\vb{M}}}\) and \(\braket{\vec{m}:\vb{M}}{\phi'} \simeq \ket{y :
        Q_u}_u \otimes \ket{\psi_{\vec{m}:\vb{M}}}\).
    \end{enumerate}

  In the latter case, make some arbitrary choice of measurements \(\vb{M} : R \to
  U(\mathcal{H})\), and expand \(\ket{\phi}\) in their common eigenbases:
  \begin{equation}
    \ket{\phi} = \sum_{n \in \F} \sum_{\vec{a}
      \in \F^R} c(n,\vec{a}) \ket{n:Q_u}_u \otimes \ket{\vec{a}:\vb{M}}_R \otimes \ket{\phi(a)}.
  \end{equation}
  Then in particular, we have that for any choice \(\vec{m} \in \F^R\),
  \begin{equation}
    \bra{\vec{m}:\vb{M}}\ket{\phi} = \sum_{n \in \F}
    c(n,\vec{m}) \ket{n:Q_u}_u \otimes \ket{\phi(n,\vec{m})}
    \simeq \ket{x:Q_u}_u \otimes \ket{\psi_{\vec{m}:\vb{M}}},
  \end{equation}
  which implies that \(c(n,\vec{m}) = 0\) whenever \(n \neq x\), and we have
  \(\ket{\phi} = \ket{x:Q_u}\otimes\ket{\psi_x}\), where
  \begin{equation}
    \ket{\psi_x} = \sum_{\vec{m} \in \F^R} c(x,\vec{m}) \ket{\vec{m}:\vb{M}} \otimes \ket{\phi(x,\vec{m})}. 
  \end{equation}
  Similarly \(\ip{\vec{m}:\vb{M}}{\phi'} \simeq \ket{y:Q_u} \otimes \ket{\psi'_y}\).
  It follows that for any \(\vec{m} \in \F^R\), we must have
  \(\ip{\vec{m}:\vb{M}}{\psi_x} \simeq \ket{\psi_{\vec{m}:\vb{M}}}\) and \(\ip{\vec{m}:\vb{M}}{\psi'_y} \simeq
  \ket{\psi_{\vec{m}:\vb{M}}}\), so that \(\ip{\vec{m}:\vb{M}}{\psi_x} \simeq
  \ip{\vec{m}:\vb{M}}{\psi'_y}\).
  Then, by the induction hypothesis, there is \(L \subseteq R\) and
  \(\vec{x},\vec{y} \in \F^{L \cup \{u\}}\) such that \(\ket{\phi} =
  \ket{\psi} \bigotimes_{v \in L \cup \{u\}} \ket{x_u:Q_u}\) and \(\ket{\phi'} =
  \ket{\psi} \bigotimes_{v \in L \cup \{u\}} \ket{y_u:Q_u}\), and we are done.
\end{proof}

\section{Proof of lemmas \ref{lem:final_layer}-\ref{lem:recursive_decomposition}}
\label{app:maximal_delay_lemmas}
\LemFinalLayer*
\begin{proof}
  Let \(A \coloneqq O \cup \{u \in V \mid (\forall v \in V): G_{uv} = 0\}\), and
  define a layer decomposition \(\Lambda'\) on \((G,I,O,\lambda)\) by
  \begin{equation}
    \Lambda'_k \coloneqq \Lambda_k \setminus A \qfor k > 0 \qand \Lambda'_0 = \Lambda_0 \cup A.
  \end{equation}
  Then it is clear that \(\Lambda'\) is more delayed than \(\Lambda\).
  Let \(C'\) be the matrix obtained by replacing, for every isolated vertex
  \(u \in V\), the \(u\)-th column of \(C\) by \(C_{uu}1_{\{u\}}\).

  We show that \((C',\Lambda')\) is an \Fflow for \((G,I,O,\lambda)\).
  \begin{enumerate}
  \item We haven't touched the diagonal elements of \(C\) and have only
    changed the columns corresponding to isolated vertices. Then
    \begin{equation}
      (GC')_{uu} =
      \begin{cases}
        \sum_{v} G_{uv} C_{vu} = 0 \qif u \qq{is isolated;} \\
        (GC)_{uu} \qq{otherwise}.
      \end{cases}
    \end{equation}
    and condition (i) of the definition is still verified.
  \item Since \(C'_{uv} = C_{uv}\) if \(u \in I\) or \(v \in O\), we have
    condition (ii) of the definition.
  \item For every \(m>n \in \N^*\),
    \(C[\Lambda'_m,\Lambda'_n] = (GC)[\Lambda'_m,\Lambda'_n] = 0\), since they
    are submatrices of \(C[\Lambda_m,\Lambda_n] = (GC)[\Lambda_m,\Lambda_n] =
    0\), and \(C[\Lambda'_m,\Lambda'_m], (GC)[\Lambda'_m,\Lambda'_m]\) are
    diagonal for the same reason. Also,
    \begin{equation}
    (GC')[\Lambda'_m,\Lambda'_0]_{uv} =
    \begin{cases}
      0 \qif v \in \Lambda_0 \qq{since otherwise \((C,\Lambda)\) is not an \Fflow;} \\
      \sum_{k \in V} G_{uk}C'_{kv} = G_{uu} C_{uu} = 0 \qif v \text{ is isolated}; \\
      \sum_{k \in V} G_{uk}C'_{kv} = \sum_{k \in V} G_{uk}C_{kv} = 0 \qif v \in O. 
    \end{cases}
    \end{equation}
    Finally, it is clear that \(C'[\Lambda'_0,\Lambda'_0]\) is diagonal if
    \(C[\Lambda_0,\Lambda_0]\) was, since we have only added zero for outputs or
    ``diagonal'' columns for isolated vertices.
    Therefore we have condition (iii).
  \end{enumerate}
  As a result, \((C',\Lambda')\) is an \Fflow for \((G,I,O,\lambda)\) that is
  more delayed than \((C,\Lambda)\).
  This implies that we must have \(A \subseteq \Phi_0\) if \(\Phi\) is maximally
  delayed.

  Now assume there is some \(v \in \Lambda_0 \setminus O\).
  We know that \(C[\Lambda_0,\Lambda_0]\) is diagonal and that
  \(GC[\Lambda_n,\Lambda_0] = \sum_{n \leqslant
    \abs{\Lambda}} G[\Lambda_n,\Lambda_k]C[\Lambda_k,\Lambda_0] =
  G[\Lambda_n,\Lambda_0]C[\Lambda_0,\Lambda_0]\).
  If \(C_{uu} \neq 0\), then for \(GC[\Lambda_0,\Lambda_0]\) to be diagonal and
  \(GC[\Lambda_n,\Lambda_0] = 0\), we must have either \(C_{uv} = 0\) or for all
  \(u \in V\), \(G_{uv} = 0\) since then \((GC)_{uv} = G_{uv}C_{vv}\) must be
  \(0\) if \(u \neq v\).
  In the latter case, \(v\) is isolated in the graph \(G\).

  In the former case, \(G_{uu}C_{uu} = 0\), and we have
  \((C_{uu},(GC)_{uu}) = (0,0)\).
  But since \(u\) is not an output, we must have \((C_{uu},(GC)_{uu}) =
  \lambda(u)\), so that \((C,\Lambda)\) is not a \Fflow for
  \((G,I,O,\lambda)\).
  As a result, there can be no such \(u\) if \((C,\Lambda)\) is a valid \Fflow.
  We conclude that \(\Lambda_0 = O \cup \{u \in V \mid (\forall v \in V): G_{uv}
  = 0\}\).
\end{proof}

\LemPenultimateLayer*
\begin{proof}
  Let \((D,\Phi)\) be a maximally delayed \Fflow for \((G,I,O,\lambda)\) and
  define \(c^u\) as the \(u\)-th column of \(D\). The only elements below the
  diagonal in column \(v \in \Phi_1\) of \(D\) correspond to \(\Phi_1\) or \(\Phi_0\).
  Since \(D[\Phi_1,\Phi_1]\) and \((GD)[\Phi_1,\Phi_1]\) are diagonal, and
  \(\Phi_0 = O\) by lemma \ref{lem:final_layer}, for any \(v \notin O \cup
  \{u\}\) we must have \(D_{vu} = c^u_v = 0\) and \((GD)_{vu} = (Gc^u)_v = 0\).
  The condition \(\lambda(u) = (c_u,(Gc)_u)\) itself corresponds to part (iii)
  of the definition of \Fflow. As a result, every maximally delayed \Fflow of
  \((G,I,O,\lambda)\) must verify equation \eqref{eq:penultimate_layer}, and
  there can be no layer decomposition \(\Phi\) where \(\Phi_1\) is not contained
  in \(\Lambda_1\).
  
  Now, assume \((G,I,O,\lambda)\) is an open graph with \Fflow, that
  \((D,\Phi)\) is a maximally delayed \Fflow and let \(u \in \Lambda_1 \setminus
  \Phi_1\). Let \(E\) be the matrix obtained by replacing the u-th column of \(D\)
  by \(c^u\) and permuting the \(u\)-th column to the start of \(\Phi_1\). Then
  \((E,\Psi)\) where
  \begin{equation}
    \Psi_k \coloneqq
    \begin{cases}
      \Lambda_1 \cup \{u\} \quad \text{if} \quad k = 1; \\
      \Lambda_k \setminus \{u\} \quad \text{otherwise};
    \end{cases}
  \end{equation}
  is a more delayed \Fflow than \((D,\Phi)\). As a result, there can be no
  such \(u\), so that if \((D,\Phi)\) is maximally delayed, \(\Phi_1 = \Lambda_1\).
\end{proof}

\LemRecursiveDecomposition*
\begin{proof}
  It is clear that \((D,\Phi)\) is a layer decomposition, since if it were not,
  this would imply that \((C,\Phi)\) is not either.
  
  There cannot be a more delayed \Fflow of \((G,I,O \cup \Lambda_1,\lambda)\) since
  that would immediately imply that there is a layer decomposition of
  \((G,I,O,\lambda)\) that is more delayed than \((C,\Lambda)\).
\end{proof}

\end{document}